\documentclass[
reprint,
superscriptaddress,
 amsmath,amssymb,
 aps,
]{revtex4-1}

\usepackage{mathtools}
\usepackage{graphicx}
\usepackage{dcolumn}
\usepackage{bm}
\usepackage{physics}
\usepackage{amsmath}
\usepackage{amssymb}
\usepackage{amsthm}
\usepackage{algorithm}
\usepackage{algorithmic}
\usepackage{booktabs} 
\usepackage{mathrsfs}
\usepackage{makecell}
\usepackage{braket}
\usepackage{xcolor}
\usepackage[roman]{complexity}
\usepackage{tikz}
\usetikzlibrary{quantikz2}
\usepackage{pdfpages}
\usepackage{pgffor}
\usepackage{caption}
\usepackage{subcaption}

\newtheorem{theorem}{Theorem}
\newtheorem{lemma}{Lemma}

\newtheorem{definition}{Definition}

\newtheorem{statement}{Statement}

\makeatletter
\AtBeginDocument{\let\LS@rot\@undefined}
\makeatother

\begin{document}

\title{Towards identifying possible fault-tolerant advantage of quantum linear system algorithms in terms of space, time and energy}
\author{Yue Tu}
\affiliation{Department of Computer Science, The University of Pittsburgh, Pittsburgh, PA 15260, USA}

\affiliation{Department of Computational and Applied Mathematics, The University of Chicago, Chicago, IL 60637, USA}

\author{Mark Dubynskyi}

\affiliation{Mason Experimental Geometry Lab, Fairfax, VA 22030, USA}

\author{Mohammadhossein Mohammadisiahroudi}
\affiliation{Department of Industrial and Systems Engineering, Lehigh University, Bethlehem, PA 18015, USA}

\author{Ekaterina Riashchentceva}

\affiliation{The University of Portland, Portland, OR 97203, USA}

\author{Jinglei Cheng}

\affiliation{Department of Computer Science, The University of Pittsburgh, Pittsburgh, PA 15260, USA}

\author{Dmitry Ryashchentsev}

\affiliation{Charles Schwab Corporation, Westlake, TX 76262, USA}

\author{Tam\'{a}s Terlaky}

\affiliation{Department of Industrial and Systems Engineering, Lehigh University, Bethlehem, PA 18015, USA}

\author{Junyu Liu}
\affiliation{Department of Computer Science, The University of Pittsburgh, Pittsburgh, PA 15260, USA}

\date{\today}

\maketitle

\noindent\textbf{Quantum computing, a prominent non-Von Neumann paradigm beyond Moore's law, can offer superpolynomial speedups for certain problems. Yet its advantages in efficiency for tasks like machine learning remain under investigation, and quantum noise complicates resource estimations and classical comparisons. We provide a detailed estimation of space, time, and energy resources for fault-tolerant superconducting devices running the Harrow-Hassidim-Lloyd (HHL) algorithm, a quantum linear system solver relevant to linear algebra and machine learning. Excluding memory and data transfer, possible quantum advantages over the classical conjugate gradient method could emerge at $N \approx 2^{33} \sim 2^{48}$ or even lower, requiring ${O}(10^5)$ physical qubits, ${O}(10^{12}\sim10^{13})$ Joules, and ${O}(10^6)$ seconds under surface code fault-tolerance with three types of magic state distillation (15-1, 116-12, 225-1). Key parameters include condition number, sparsity, and precision $\kappa, s\approx{O}(10\sim100)$, $\epsilon\sim0.01$, and physical error $10^{-5}$. Our resource estimator adjusts $N, \kappa, s, \epsilon$, providing a map of quantum-classical boundaries and revealing where a practical quantum advantage may arise. Our work quantitatively determine how advanced a fault-tolerant quantum computer should be to achieve possible, significant benefits on problems related to real-world.}

\section{Introduction}\label{sec:introduction}
Quantum computing \cite{nielsen2010quantum} stands as the leading paradigm of next-generation computing technology. Unlike traditional computing, quantum computers can access and manipulate quantum states, which serve as carriers of information for computational purposes. One of the primary motivations for developing quantum computing is its potential to achieve a \emph{possible quantum advantage} over its classical counterparts in solving certain problems. This potential quantum advantage, studied primarily in terms of time cost within computational complexity theory, has been theoretically demonstrated in problems such as factoring \cite{shor1999polynomial}, searching \cite{grover1996fast}, and simulation \cite{lloyd1996universal}. These advancements offer hope for sustaining or surpassing Moore's law in the semiconductor industry.

However, beyond the time complexity estimated in the gate models from theoretical computer science, it is challenging to estimate and justify the possible quantum advantage in practice. First, practical cost estimation of quantum computing requires state-of-the-art knowledge, from detailed theory covering prefactors in front of the big-$O$ notation of complexity \cite{jennings2023efficient,dalzell2023end}, to explicit designs of quantum hardware, and it includes more comprehensive measurement such as time costs (measured in seconds), space costs (number of physical qubits), and in particular energy costs. Especially, the complex nature of quantitative energy efficiency estimation is highly uninvestigated, although the possible energy advantage of quantum computing algorithm is discussed mostly in the qualitative argument \cite{auffeves2022quantum,scholten2024assessing,liu2023potential}. Second, although the existence of potential quantum advantage for some algorithms is solidly justified in theory, it is challenging to prove that those algorithms can turn into real-world, significant benefits especially for commercial applications \cite{FASQ}. Finally, quantum states are extremely fragile and current quantum processors are noisy, making quantum error correction the only possible way to make large-scale, fault-tolerant quantum computing. Fault-tolerance, although sustainable in theory, requires lots of additional resources and experimental challenges, making precise resource estimation much more challenging.

In this work, we address those challenges by conducting a full-stack, energy-aware resource estimation for the so-called Harrow-Hassidim-Lloyd (HHL) algorithm \cite{harrow2009quantum}. The HHL algorithm provides a Quantum Linear System Algorithm (QLSA) that can be used for solving linear algebra problems. Given a linear equation $ A\ket{x} = \ket{b} $, the algorithm returns a quantum state $ \ket{x} = A^{-1} \ket{b} $ as a solution. For certain classes of matrices, it has been theoretically shown that the algorithm runs in $ \text{poly} (\log N) $ time for an $ N \times N $ matrix, making it exponentially faster than any known classical counterpart. Complexity-theoretic argument also suggests that the algorithm under certain settings are BQP-complete \cite{harrow2009quantum}. Since linear algebra applications are ubiquitous in modern science and technology, the algorithm is regarded as one of the most promising applications for large-scale, fault-tolerant quantum processors \cite{dervovic2018quantum}.

To identify the regime of possible, practical quantum advantages, we systematically investigate the costs of QLSA. The primary input of QLSA is a general Hermitian matrix $A$ of dimension $N = 2^n$, with row sparsity $s$ (where $s = \max_j ( \text{ \# of nonzero entries in row } j \text{ of } A)$), precision $\epsilon$, and condition number $\kappa$. For non-Hermitian matrices, one can effectively solve a corresponding Hermitian matrix instead, requiring only one additional qubit \cite{harrow2009quantum,dervovic2018quantum}. At the logical level, we primarily adopt the original and classic HHL framework \cite{harrow2009quantum,dervovic2018quantum}. For quantum simulation, we employ the one-sparse Hamiltonian simulation model to implement $e^{i A t}$, using \cite{childs2021theory} to estimate the Trotter error. Finally, we provide an explicit formula for estimating resource requirements at the logical level for HHL, particularly focusing on the number of different types of quantum gates. These estimates are validated using a quantum circuit state-vector simulation code written in Q\# \cite{qsharp}.

Moreover, we extend our analysis to resource estimations at the fault-tolerant level. To achieve this, we utilize the framework proposed by Litinski \cite{litinski2019game} to estimate the fault-tolerance resource. The framework assume Clifford gates are cheap and only accepts the number of $T$-gates as input, then search for the optimal surface code and magic state distillation scheme (we primarily consider 15-1, 116-12, 225-1) scheme that will minimize the space-time cost with logical error rate below a certain threshold (in this work we primarily set the logical error around $0.01$, around the same range of the precision $\epsilon$, and a superconducting device with physical error $10^{-5}$). Finally, we compare our fault-tolerant resource estimates with those of classical algorithms. In particular, we discuss the classical counterpart to the Quantum Linear System Algorithm (QLSA): the conjugate gradient (CG) method \cite{hestenes1952methods}. The CG method is widely regarded as the classical analogue of QLSA, as it similarly exploits the sparsity and well-conditioned nature of the problem. Its time complexity is given by $O(N s\kappa \log(1/\epsilon))$. We conduct a detailed analysis of its exact scaling behavior and provide an estimation of its clock cycle requirements.

For superconducting devices, the primary energy costs associated with executing the HHL algorithm stem from the cooling system. On the quantum side, we estimate the energy consumption based on the model proposed in \cite{martin2021energyusequantumdata}, assuming that the energy cost scales linearly with the number of qubits, the total energy cost is given by the product of the runtime, the number of physical qubits, and a cooling efficiency factor. To estimate the cooling efficiency, we refer to \cite{parker2023estimating}, which provides an approximation based on IBM Quantum System Two's dilution refrigerator, with an estimated energy efficiency of 6.25 Watts per qubit. However, it is important to note that cooling efficiency is highly dependent on the quantum computer’s architecture and may vary significantly as quantum hardware evolves. Consequently, this estimation should be viewed as a rough approximation rather than a precise metric.

On the classical side, we estimate the energy consumption using the power requirements of a typical desktop CPU. A standard desktop processor operating at a clock frequency of 1 GHz consumes approximately 50 Watts. This provides a baseline for evaluating the comparative energy efficiency of classical and quantum computations in the context of the QLSA algorithm \cite{cpupower}. In Appendix, we discuss how to make use of parallel cluster computing methods to perform similar algorithms, where more precise number could be found in the Top500 list \cite{top500_2024} for the best classical devices. 

Note that we consider only the costs incurred during the computation process. The costs associated with the interface between classical and quantum processors—such as uploading the matrix $A$ and the vector $\ket{b}$ to the quantum device, as well as downloading the state $\ket{x}$ to classical memory—are not addressed in this work. The uploading problem can be mitigated by fast quantum memory solutions, such as Quantum Random Access Memory (QRAM) \cite{giovannetti2008quantum}, whose circuits are, in fact, classically simulatable \cite{hann2021resilience}. Meanwhile, the downloading process depends on the specific needs of the user. Since classical users cannot handle exponentially large datasets \cite{aaronson2015read}, quantum tomography techniques \cite{huang2020predicting,liu2024towards} can be employed to extract partial information from $\ket{x}$.

Our work is organized as follows. In Section \ref{sec:results}, we discuss our results comparing quantum and classical approaches, along with their implications for computational speed and energy efficiency. In Section \ref{sec:methods}, we outline the key methodologies employed to obtain these results. At the logical level, this includes circuit implementation and exact scaling derivation, while at the physical level, it encompasses the implementation of the surface code optimizer. In Section \ref{sec:concoutlook}, we provide conclusions and an outlook on future research directions. Additional technical details, including a detailed formulation of the HHL algorithm, Trotter simulation strategies, surface code setups, classical counterparts, and energy analysis, are provided in the Appendix.

\section{Results}\label{sec:results}

In this chapter, we present the main theoretical and numerical results of our study. For the HHL resource estimation, we primarily use the following statement, 

\begin{statement}\label{state:logicalscale}
  Consider a system of linear equations $Ax = b$, where $ A $ has been scaled such that its eigenvalues lie in the interval $[\frac{1}{\kappa}, 1]$. Denote by $ |b\rangle $ the normalized quantum state proportional to $ b $. Assume access to an oracle that provides access to the elements of a sparse submatrix of $ A $. Then, there exists a quantum algorithm that takes the input state $ |b\rangle $ and outputs the normalized solution state $ |\tilde{x}\rangle $ with additive error $ || \tilde{x} - x || $ less than $ \epsilon $. The algorithm requires $ T $ $T$-gates and $ Q $ queries to the oracle, where:
  \begin{align}
    T &\lesssim  \frac{\sqrt{\frac{320}{3}} \pi \kappa^2 s}{\epsilon^2} \times (18 \log N + 90r + 15)~, \\
    Q &\lesssim  \frac{\sqrt{\frac{320}{3}} \pi \kappa^2 s}{\epsilon^2} \times 2~, \label{eq:logicalquery}
  \end{align}
  where $ N $ is the matrix size, $ \kappa $ is the condition number, $ s $ is the sparsity of $ A $, and $ \epsilon $ is the precision.
\end{statement}
The statement could be used as a reasonable, saturated resource estimation, and is semi-rigorous based on a combination of mathematical proofs and numerical experiments, which has been discussed in detail in Appendix. Compared with the original HHL algorithm with $O \left(\log(N) \cdot s^2 \kappa^2 / \epsilon\right)$, polynomial scalings on $s,\kappa,1/\epsilon$ has been updated for the convenience of exact pre-factor derivation. Thus, Statement \ref{state:logicalscale} examines the operation count overhead associated with the algorithm's execution independent of hardware parameters, which has been turned into fault-tolerant computing and energy resources explained in Section \ref{sec:methods}. 

In the classical side we have a similar statement evaluated based on the number of floating-point-operations (FLOPs). The detailed derivation of Statement \ref{statement:classical} we refer readers to Appendix.

\begin{statement}\label{statement:classical}
    For the CG method solving the linear system $Ax = b$, the algorithm runs in $C$ FLOPs, where:
  \begin{align}
    C & \lesssim  (4Ns + 14N) \times \left( \frac{1}{2} \kappa \log\left(\frac{2}{\epsilon}\right) \right)~.
  \end{align}
  Here, $N$ is the matrix size, $\kappa$ is the condition number, $s$ is the sparsity, and $\epsilon$ is the precision.  
\end{statement}
Finally, we consider hardware parameters to illustrate the specific runtime overhead and its relationship with various input parameters. The superconducting quantum device is assumed to have physical qubit error rate of $10^{-5}$, a logical cycle execution time of 10 ns, and an energy consumption of 6.25 Watts per physical qubit. For the classical side, we assume 50 Watts per GHz has been costed with 1 GHz frequency. With error correction and magic state distillation cycles described in Section \ref{sec:methods}, We present heat maps with different $N, s, \kappa,\epsilon$ in Figure \ref{fig:big_figure} to compare the difference between quantum and classical computing. 

In these heat maps, blue regions indicate a ratio $\frac{\text{Classical Overhead}}{\text{Quantum Overhead}} \ge 1$, implying that the quantum algorithm incurs lower overhead than its classical counterpart. Consequently, these regions can be interpreted as exhibiting possible quantum advantage. Also, In Figure \ref{fig:big_figure}, we depict how runtime, space, and energy costs vary with matrix size $N$, where we assume certain conditions for $s,\kappa,\epsilon$. Note that $N,s,\kappa,\epsilon$ are all changable in our program, and the assumption may not perfectly align with future quantum computers. For readers interested in different configurations, the runtime estimation and energy consumption can be easily rescaled based on the ratio of the new parameters for clock cycle time and energy consumption per qubit. 

As a summary, we have the following primary conclusions:
\begin{itemize}
\item From Figure \ref{fig:big_figure} (a,b), we see a precise boundary in the space of key algorithm parameters $N,\kappa,s,1/\epsilon$ between $\frac{\text{Classical Overhead}}{\text{Quantum Overhead}} \le 1$ and $\frac{\text{Classical Overhead}}{\text{Quantum Overhead}} \ge 1$, indicating a boundary of potential, practical quantum advantage. 
\item Based on the crossing point of Figure (c,d), we see that possible quantum advantages could emerge around the matrix size $N \approx 2^{33} \sim 2^{48}$ or even lower, requiring ${O}(10^5)$ physical qubits, ${O}(10^{12} \sim 10^{13})$ Joules of energy, and approximately 10 days of computational time (${O}(10^{6})$ seconds). Note that the quantum energy advantage might arrive later than the computational advantage, since current superconducting quantum devices are relatively costly compared to classical ones in terms of energy due to their dilution refrigerator cooling systems. More advanced cryogenic technologies and better quantum device designs based on superconducting hardware or beyond might be more energy-efficient \cite{jaschke2023quantum}.
\item Polynomial factor about $s,\kappa,\epsilon$ in practical setups might matter and change the boundary of practical quantum advantage. For instance, since in our version of the algorithm, the dependence on $s$ is light both in classical and quantum algorithms, so the boundary in the $s$ direction might be flatten. However, practical quantum advantage might always emerge for large enough $N$ due to the $\log N$ scaling.  
\end{itemize}

\begin{figure*}
    \centering
    \begin{subfigure}{0.94\textwidth}
        \centering
        \includegraphics[width=\textwidth]{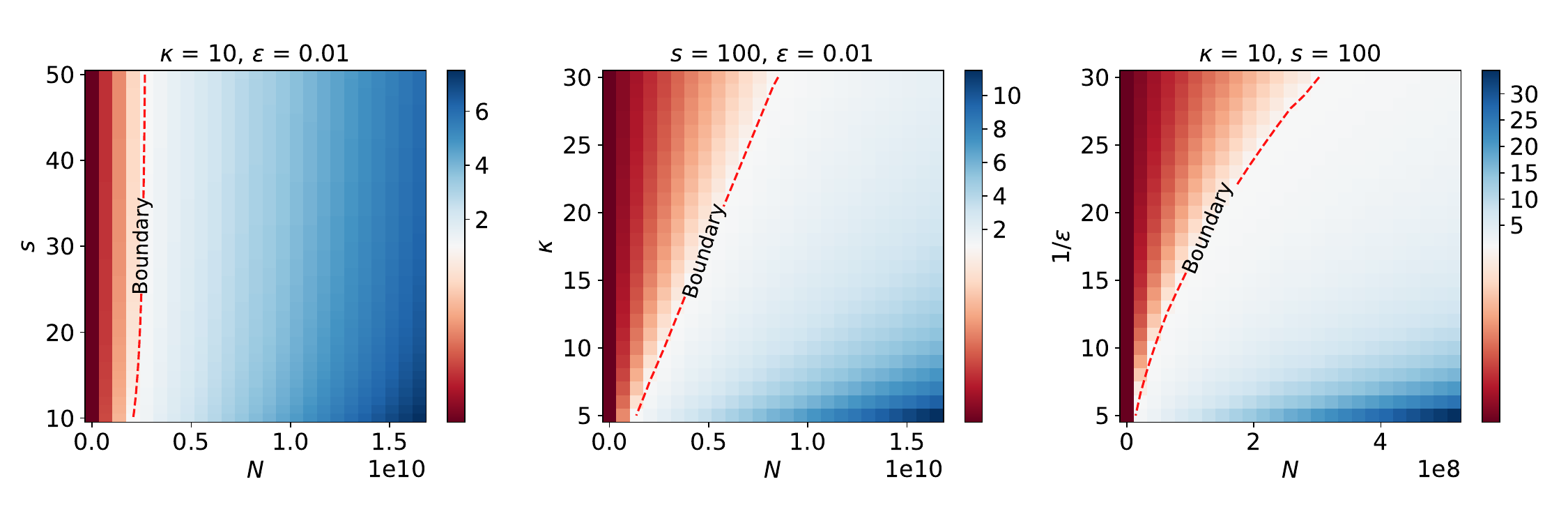}
        \caption{Runtime Ratio: $\frac{\text{Classical Overhead}}{\text{Quantum Overhead}}$, the blue region corresponds to possible quantum advantage.}
    \end{subfigure}
    \begin{subfigure}{0.94\textwidth}
        \centering
        \includegraphics[width=\textwidth]{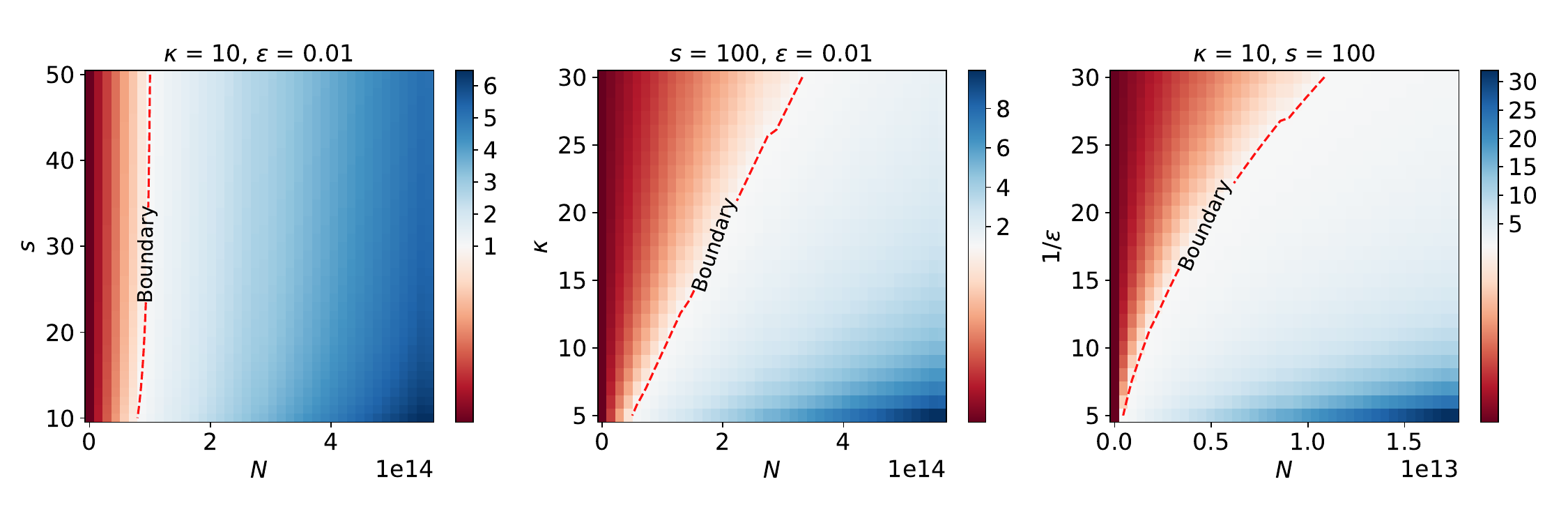}
        \caption{Energy Ratio, $\frac{\text{Classical Overhead}}{\text{Quantum Overhead}}$, the blue region corresponds to possible quantum advantage.}
    \end{subfigure}
    \begin{subfigure}{0.3\textwidth}
        \centering
        \includegraphics[width=\textwidth]{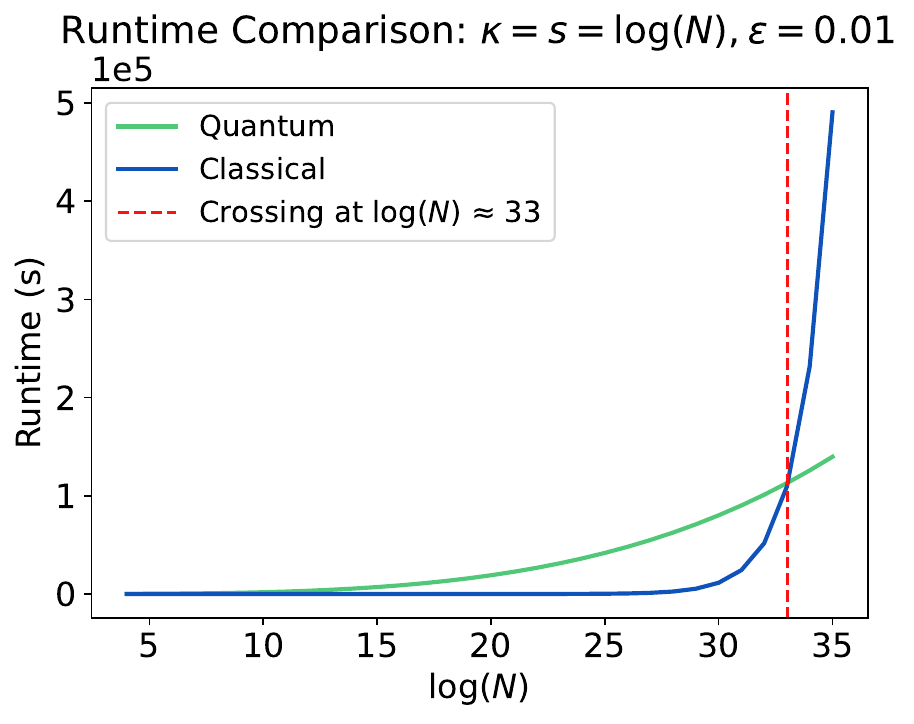}
        \caption{Runtime Comparison.}
    \end{subfigure}
    \begin{subfigure}{0.3\textwidth}
        \centering
        \includegraphics[width=\textwidth]{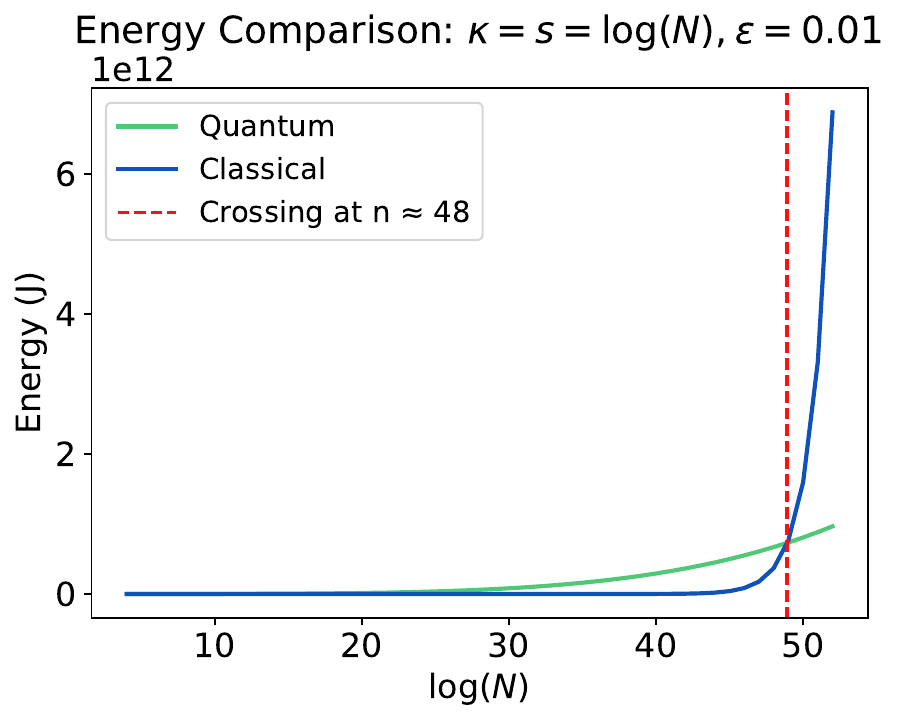}
        \caption{Energy Comparison.}
    \end{subfigure}
    \begin{subfigure}{0.3\textwidth}
        \centering
        \includegraphics[width=\textwidth]{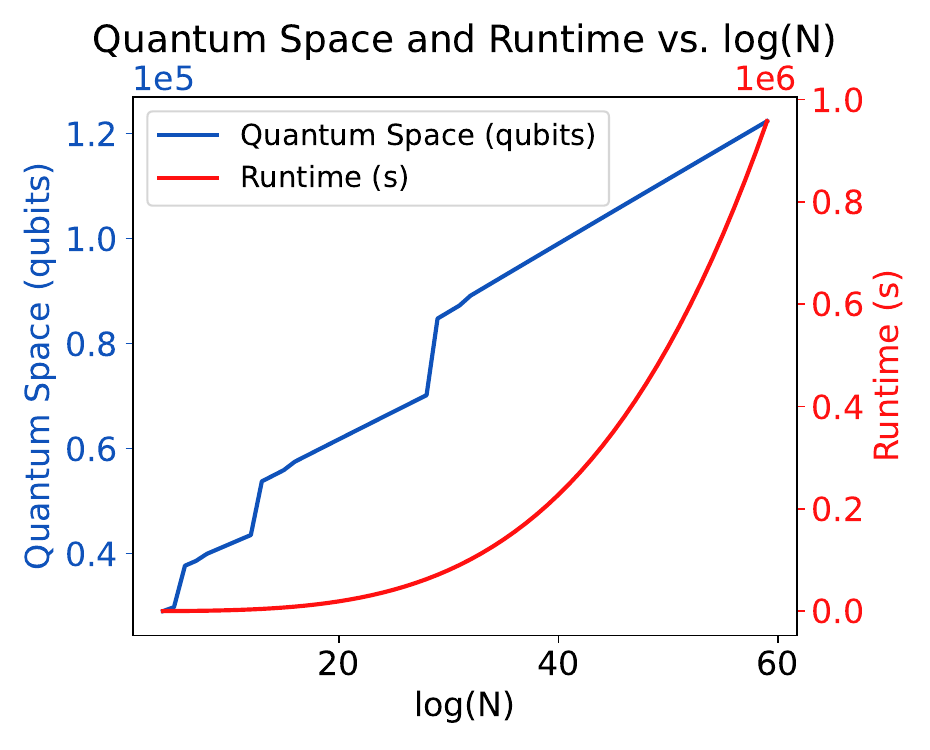}
        \caption{Runtime and Space Cost.}
    \end{subfigure}

    \caption{
    \textbf{Comparison of Quantum Linear System Algorithm (QLSA) and Classical Conjugate Gradient Algorithm.} In the plots, $N$ represents matrix size, $s$ is the sparsity, $\kappa$ is the condition number, $\epsilon$ is the precision tolerance in terms of vector-2 norm.
    (a) \& (b): Heatmaps showing the ratio of classical runtime to quantum runtime and classical energy consumption to quantum energy consumption under different problem parameters. Blue regions indicate a possible quantum advantage.
    (c) \& (d): Runtime and energy comparison between quantum and classical algorithms, assuming condition number and sparsity scale as $\log (N)$, and $\epsilon$ to be 0.01.
    (e): Runtime and qubit count scaling for the quantum algorithm over larger matrix sizes.
    }
    \label{fig:big_figure}
\end{figure*}

\section{Methods}\label{sec:methods}
The quantum resource estimation primarly contain logical resource analysis and physical resource analysis.

At the logical level, we first implemented the quantum circuit for the Quantum Linear Systems Algorithm (QLSA) using Q\#, a Microsoft programming language designed for circuit-based quantum computing \cite{qsharp}. This implementation allowed us to estimate the logical qubit and $T$-gate requirements for each submodule. Among all submodules in the QLSA algorithm, the primary bottleneck lies in one specific subroutine—the one-sparse Hamiltonian simulation (HS1). This subroutine is responsible for constructing a unitary operation that implements $e^{iAt}$, where $A$ is the input matrix. The high resource demand of HS1 stems from the Trotterization process, which requires applying HS1 repeatedly to suppress the error to a desired level. Additionally, the unitary operation itself must be executed multiple times as part of the Quantum Phase Estimation (QPE) procedure, further increasing the total number of HS1 invocations. To precisely determine the scaling behavior of HS1, it is essential to analyze the algorithm's error. There are two primary sources of error in the QLSA algorithm: (1) errors introduced by Trotterization and (2) errors from the Quantum Fourier Transform (QFT). While previous analyses of QLSA have not explicitly addressed Trotterization, our work provides the first comprehensive error analysis that accounts for both error sources simultaneously. Through mathematical derivation, we obtained an exact upper bound for the total error and found that this bound is significantly higher than the average-case error. To refine this analysis, we conducted numerical experiments to give a probabilistic bound, revealing that the expected cost of the HHL algorithm is substantially lower than its worst-case estimate. These refinements led to the derivation of Statement \ref{state:logicalscale}. Based on our error estimation, we determined the required number of HS1 repetitions for specific input cases. The gate count per HS1 operation was obtained from our code implementation. Since $T$-gates represent the primary resource bottleneck in our analysis, we focus on tracking only the $T$-gate count, although other gates are also estimated in Appendix. The total number of $T$-gates required is given by the product of the total number of HS1 invocations and the $T$-gate count per HS1 operation, which gives us the result as shown in Statement \ref{state:logicalscale}.

After determining the number of logical qubits and logical $T$-gates required by the algorithm, the next step is to translate these logical resources into the fault-tolerant physical resources under a quantum error correction (QEC) scheme. Among the various existing QEC schemes, we adopt the surface code in this study as it has the highest fault-tolerant error threshold and has been demonstrated by recent experiments \cite{acharya2024quantumerrorcorrectionsurface}. In essence, when executing quantum algorithms using the surface code model, physical qubits are categorized into data blocks and distillation blocks. Data blocks are responsible for storing logical qubits and consuming the so-called \emph{magic states}, while distillation blocks generate these magic states. Each time a data block consumes a magic state produced by a distillation block, a $T$-gate operation is executed. Consequently, the total physical qubit count is simply the sum of physical qubits in the data and distillation blocks. The runtime of surface code-based quantum computing is approximately given by the number of $T$-gates divided by their execution speed, which is determined by the specific configuration of data and distillation blocks as well as the hardware parameters of the quantum computer. To determine the physical qubit count and runtime, it is essential to establish an optimal configuration of data and distillation blocks. Our goal is to find a configuration that minimizes the space-time volume, defined as the product of the number of physical qubits and the runtime. Achieving this requires considering several constraints and trade-offs, which are listed as follows: 1). The distillation block must be large enough to achieve the desired $T$-state error rate. 2). The code distance of the surface code must be large enough to achieve the desired logical error rate. 3). The data block's efficiency to consume magic states comes at a price of using more physicla qubits. 4). The distillation block's efficiency to produce magic states comes at a price of using more physical qubits.

Moreover, we use the procedure proposed in \cite{litinski2019game} to optimize space-time volume under such considerations. The procedure includes steps: 1). Determine the distillation protocol. We select the most cost-effective distillation protocol that achieves the required $T$-state error rate. 2). Construct a minimal setup. We design the space-optimal protocol. 3). Determine the code distance. We identify the minimum code distance that ensures the error rate remains below the threshold. 4). Add distillation blocks. Building on the space-optimal protocol, we incrementally adopt data block and distillation block protocols with larger space overhead but reduced time overhead until the optimal space-time cost is achieved.

We implemented these four steps in a resource simulator, which will be further discussed in the Appendix. By inputting the algorithm's logical $T$-gate count, logical qubit count, and reasonably assumed quantum computer parameters (including quantum gate fidelity, the number of physical qubits, and clock frequency), the estimator identifies the configuration of distillation and data blocks that minimizes the space-time volume. It then outputs the corresponding runtime and physical qubit overhead. We are now left with energy estimation. We use the model proposed in \cite{martin2021energyusequantumdata}, where the primary energy cost of quantum computers is cooling, and the energy consumption grows with the number of qubits because a single refrigeration unit can accommodate only a limited number of superconducting qubits. As a result, the quantum computer's power consumption is proportional to the number of physical qubits. We define energy efficiency as the power consumption per qubit and use the data provided by \cite{parker2023estimating}, which targets IBM Quantum System Two's dilution refrigerator. This system houses 4158 physical qubits and operates at a power of 27 kW, resulting in a power efficiency of 6.25 Watts per qubit. These results provide an intuitive and detailed reference for end-to-end runtime estimation of the resources required to run the QLSA algorithm and serve as a practical basis for developing QLSA-related applications. However, we emphasize that current quantum computers are far from being capable of running the HHL algorithm, and future quantum hardware may differ significantly from existing systems. Additionally, there has been significant progress in the QLSA algorithm, meaning that the algorithm we are estimating may not be optimal. Therefore, the specific results we provide should be viewed as rough baseline predictions based on the current state of quantum computing technology.

In the classical part, we mainly address the conjugate gradient (CG) method, while an alternative method, the Cholesky decomposition (CD) method, is discussed in the Appendix. We primarily focus on the former, as CG is generally considered the classical counterpart of QLSA due to its iterative nature and ability to handle sparse and well-conditioned systems. On the other hand, the CD method is not a direct counterpart to QLSA, as it is deterministic and does not depend on $\epsilon$. However, CD is more commonly used in supercomputers, and we include it in the Appendix, where all resource costs are explicitly derived.

\section{Conclusion and outlook}\label{sec:concoutlook}
In this work, we perform an end-to-end resource estimation of the HHL algorithm in terms of time, space, and energy. Unlike existing resource estimations, including \cite{jennings2023efficient}, our work provides a detailed analysis of fault tolerance based on surface codes. Moreover, we precisely address energy efficiency by employing physical models of superconducting quantum devices. Our findings confirm that quantum computing will indeed have an energy advantage at a large scale based on current superconducting quantum technologies, resolving a long-standing concern about the energy efficiency of quantum computing in practical applications \cite{auffeves2022quantum,scholten2024assessing}. Other innovative aspects of our work include a detailed comparison between quantum and classical approaches to solving linear systems, and an analysis of theoretical upper bounds versus practical average cases, the details of which are summarized in the Appendix.

Although we confirm that potential quantum advantages might appear at the scale of $2^{33} \sim 2^{48}$ matrix sizes and $10^5$ physical qubits, the practical costs could be lower due to the overestimation of resource counts in the Trotter simulation. Moreover, smarter designs for quantum error correction and energy-sustainable innovations in quantum hardware might further enhance performance. On the other hand, compared to other resource estimations of quantum algorithms in different domains (such as cryptography used for digital signatures and blockchains \cite{litinski2023compute}), the HHL algorithm might achieve potential quantum advantages with a smaller number of physical qubits \cite{litinski2023compute}. However, better classical and quantum designs including memories, might further push forward or backward the boundary of practical quantum advantage \cite{gouzien2021factoring}.

Our work presents significant opportunities for bringing quantum computing into practice. For example, we do not address the uploading and downloading challenges of the HHL algorithm, which may be critical for data-intensive applications. Active developments are underway on both the hardware side \cite{hann2019hardware} and the system design side \cite{weiss2024faulty,wang2024fundamental} toward large-scale and error-resilient QRAMs \cite{hann2021resilience} in promising hardware platforms such as circuit/cavity QED/QAD \cite{hann2019hardware}. Notably, since QRAM circuits are classically simulatable, precise estimates of QRAM costs can be investigated on a large scale \cite{hann2021resilience}. Moreover, for the downloading process, a detailed classification of quantum computing user needs may be essential. If quantum customers are classical, what will classical customers require? The first few elements of the state $\ket{x}$? Some expectation value of operators \cite{huang2020predicting}? Other important directions include further exploration of energy efficiency. Can we provide energy estimations for other hardware, such as neutral-atom quantum computers? Can we make device-independent statements about energy advantages? Can we generalize such estimations to a broader range of quantum algorithms? Finally, it would also be interesting to investigate more advanced versions of quantum algorithms \cite{ambainis2010variable,Childs_2017,low2024QLAOptimal,costa2021optimal,dalzell2024shortcut} and classical algorithms \cite{hestenes1952methods}, as well as improved designs for quantum error correction codes \cite{bravyi2024high,xu2024constant,viszlai2023matching} and quantum memories \cite{gouzien2021factoring}.



\textit{Acknowledgements.---}We thank Fred Chong, Jens Eisert, Tranvis Humble, and Liang Jiang for useful discussions. YT, JC and JL are supported in part by the University of Pittsburgh, School of Computing and Information, Department of Computer Science, Pitt Cyber, and the PQI Community Collaboration Awards. MM and TT are supported by Defense Advanced Research Projects Agency as part of the project W911NF2010022, {\em The Quantum Computing Revolution and Optimization: Challenges and Opportunities}. 

\bibliographystyle{unsrt}
\bibliography{bibliography.bib}

\clearpage

\pagebreak

\onecolumngrid
\appendix

\vspace{0.5in}

\begin{center}
	{\Large \bf Appendix}
\end{center}

\section{QLSA algorithm implementation}\label{sec:QLSAImp}

\subsection{A brief introduction to QLSA}\label{sec:IntroQLSA}

The Quantum Linear Systems Algorithm (QLSA) is a quantum algorithm designed to solve systems of linear equations efficiently. It relies on principles of quantum mechanics to provide a possible superpolynomial advantage for specific cases compared to classical methods.

The core idea is to represent the solution to a linear system $A x = {b}$ in the quantum state $|x\rangle$, where $A$ is a sparse and well-conditioned matrix. The algorithm involves several key steps:

\begin{itemize}
    \item \textbf{Preparation of the Input State:} The algorithm starts with the quantum encoding of the input vector $|b\rangle$ into a quantum state.
    \item \textbf{Quantum Phase Estimation (QPE):} QPE approximate the eigenvalues of the matrix $A$ by estimating the phase of the unitary $e^{iAt}$. Such unitary is implemented using the Trotter decomposition method.
    \item \textbf{Matrix Inversion:} The eigenvalues of $A$ are inverted conditionally, enabling the construction of the quantum state corresponding to the solution.
    \item \textbf{Measurement:} Finally, measurements are performed to extract information about the solution vector.
\end{itemize}

QLSA's efficiency is highly dependent on the sparsity and condition number of the matrix $A$, making it particularly useful for problems where these constraints are satisfied

\subsection{Resource bottleneck in QLSA}\label{sec:QLSABottleneck}

A significant computational bottleneck in the QLSA arises from the repetitive application of one-sparse Hamiltonian simulation, which is a fundamental component of the algorithm. This challenge emerges due to the following reasons:

\begin{itemize}
    \item \textbf{Quantum Phase Estimation (QPE):} The QPE step requires the Hamiltonian simulation of $A$, which is an $s$-sparse matrix, to be applied multiple times to achieve sufficient precision. The repetition ensures the accurate estimation of eigenvalues, which directly affects the correctness of the algorithm.
    \item \textbf{Simulation of $s$-Sparse Hamiltonians:} To implement the simulation of $s$-sparse Hamiltonians, the Trotter decomposition method is often employed. However, this approach introduces further overhead, as it decomposes the $s$-sparse Hamiltonian into a series of one-sparse Hamiltonians.
    \item \textbf{One-Sparse Hamiltonian Simulation:} Each $s$-sparse Hamiltonian simulation requires the repetitive application of one-sparse Hamiltonian simulations. This nested structure compounds the computational cost, making one-sparse Hamiltonian simulation a critical bottleneck for resource efficiency.
\end{itemize}

In essence, the hierarchical nature of the QLSA—where QPE requires multiple iterations of $s$-sparse simulations, and $s$-sparse Hamiltonian simulations depend on repeatedly applying a series of one-sparse simulations—highlights the algorithm's computational resource bottleneck.

Therefore, our resource estimation focuses solely on the one-sparse Hamiltonian simulations, reducing the overall resource estimation to (1) the number of one-sparse Hamiltonian simulations. (2) the resource cost of each one-sparse Hamiltonian simulation. Next we will discuss the resource cost of each one-sparse Hamiltonian simulation by giving
a concrete code implementation of such step. We left the analysis for number of one-sparse Hamiltonian simulation in the error analysis section.

\subsection{One-sparse Hamiltonian simulation implementation}\label{sec:QLSACircuit}
In general, directly implementing the unitary $e^{iHt}$ is challenging because $H$ is an extremely large matrix, and computing its exponential is highly complex. Here we use the methods in \cite{Berry_2006}. The core idea for implementing $e^{iHt}$ relies on the \textit{unitary conjugation theorem}, which states that for a matrix exponential of the form $e^{iUHU^\dagger}$, it is equivalent to $U e^{iH} U^\dagger$.

Thus, if we can decompose $H$ into the form $U_1 U_2 \dots H' \dots U_2^\dagger U_1^\dagger$, where each $U_i$ and $e^{iH't}$ can be efficiently implemented, we effectively achieve the implementation of $e^{iHt}$.

In the following, we first explain the decomposition process and then discuss the implementation of the resulting unitary operations.

\begin{lemma}\label{lemma:Hdecomp}
    For a one-sparse Hermitian matrix H with real valued entries, it can be decomposed into the form (we ignore the tensor with identity matrix):
    \begin{align}
        H & = M (T \otimes F) M^\dagger \label{eq:Hdecomp}
    \end{align}
    where $M$ is the oracle, T is the swap operator on two registers, and F is diagonal operator, specifically:
    \begin{align}
        M|x\rangle |0\rangle |0\rangle |0\rangle & = |x\rangle |m(x)\rangle |0\rangle |w(x)\rangle \\
        T|x\rangle |y\rangle                     & = |y\rangle |x\rangle                           \\
        F|x\rangle                               & = x|x\rangle
    \end{align}
\end{lemma}

\begin{proof}
    Suppose the x's column of $H$ has the non-zero element $w(x)$ at position $m(x)$, then we have:
    \begin{align}
        H\ket{x} = w(x)\ket{m(x)}.
    \end{align}

    Thus to construct the operator $H$, is equivalent to mimic such behavior. We can achieve this by the following steps:

    \begin{enumerate}
        \item Apply $M$ to $\ket{x}$ and the ancilla:
              \begin{align}
                  \ket{x}\ket{0}\ket{0}\ket{0} \xrightarrow{M} \ket{x}\ket{m(x)}\ket{0}\ket{w(x)}.
              \end{align}

        \item Apply $T = A^\dagger Z A$:
              \begin{align}
                  \xrightarrow{T} \ket{m(x)}\ket{x}\ket{0}\ket{w(x)}.
              \end{align}

        \item Apply $F$:
              \begin{align}
                  \xrightarrow{F} w(x)\ket{m(x)}\ket{x}\ket{0}\ket{w(x)}.
              \end{align}

        \item Apply $M^\dagger$:
              \begin{align}
                  \xrightarrow{M^\dagger} w(x)\ket{m(x)}\ket{x \oplus m(m(x))}\ket{0}\ket{w(x) \oplus w(m(x))}
                  = w(x)\ket{m(x)}\ket{0}\ket{0}\ket{0}.
              \end{align}
    \end{enumerate}
    Thus $H$ is equivalent to $M (T \otimes F) M^\dagger$.
\end{proof}

Now we are left to implement $e^{iTt}$ and $e^{iFt}$. Later, we will see that $e^{iFt}$ can be directly implemented. For $e^{iTt}$, we have to apply unitary conjugation again:

\begin{lemma}\label{lemma:Tdcomp}
    For a swap operator $T$, it can be decomposed into the form:
    \begin{align}
        T = \bigotimes_{l=1}^{n} W^{l, n+l} \bigotimes_{l=1}^{n} E^{l, n+l} \bigotimes_{l=1}^{n} W^{l, n+l} = \tilde{W}\tilde{E}\tilde{W}.
    \end{align}
    where $W$ and $E$ are two qubits operators with the following matrix representation:
    \begin{align}
        W = \begin{bmatrix}\label{eq:Wmat}
            1 & 0                  & 0                   & 0 \\
            0 & \frac{1}{\sqrt{2}} & \frac{1}{\sqrt{2}}  & 0 \\
            0 & \frac{1}{\sqrt{2}} & -\frac{1}{\sqrt{2}} & 0 \\
            0 & 0                  & 0                   & 1
        \end{bmatrix}
        \quad \quad
        E = \begin{bmatrix}
            1 & 0 & 0  & 0 \\
            0 & 1 & 0  & 0 \\
            0 & 0 & -1 & 0 \\
            0 & 0 & 0  & 1
        \end{bmatrix}.
    \end{align}
\end{lemma}
\begin{proof}
    Notice that $T$ is equivalent to the swap operator on $n$ pairs of qubits, i.e., $T = \bigotimes_{l=1}^{n} S^{l, n+l}$.
    By diagonalizing $S$, we obtain $WEW$, which completes the proof.
\end{proof}

Again the operator $\tilde{E}$ has to be decomposed, notice that the behavior of $\bigotimes_{l=1}^{n} E^{l, n+l}$ is just compute the parity of $x_iy_i$, and add a phase 1 for
even parity, and -1 for odd parity. This can be implemented by using $n$ Toffli-like gates to compute the parity and store in an
ancilla qubit, and then apply $Z$-gate to the ancilla. This gives us:

\begin{lemma}\label{lemma:tildeEdecomp}
    For $\tilde{E} = \bigotimes_{l=1}^{n} E^{l, n+l}$, it can be decomposed into the form:
    \begin{align}
        \tilde{E} = \tilde{T_f} Z \tilde{T}_f^\dagger
    \end{align}
    where $\tilde{T_f}$ are set of Toffli-like gates that is used to compute the parity, $Z$ is $Z$-gate on parity qubit.
\end{lemma}

By Lemma \ref{lemma:Hdecomp}, \ref{lemma:Tdcomp}, and \ref{lemma:tildeEdecomp}, we have the following theorem:
\begin{theorem}\label{thm:main}
    For a one-sparse Hermitian matrix $H$ with real valued entries, $e^{iHt}$ can be implemented by:
    \begin{align}
        e^{iHt} = M \tilde{W} \tilde{T_f} (e^{iZt} \otimes e^{iFt}) \tilde{T}_f^\dagger \tilde{W}^\dagger M^\dagger.
    \end{align}
\end{theorem}

\begin{proof}
    By Lemma \ref{lemma:Hdecomp}, we have:
    \begin{align}
        e^{iHt} = e^{iM (T \otimes F) M^\dagger t} = M e^{i(T \otimes F)t} M^\dagger.
    \end{align}
    By Lemma \ref{lemma:Tdcomp}, we have:
    \begin{align}
        e^{i(T \otimes F)t} = e^{i\tilde{W}\tilde{E}\tilde{W}t} = \tilde{W} e^{i\tilde{E}t} \tilde{W}^\dagger.
    \end{align}
    By Lemma \ref{lemma:tildeEdecomp}, we have:
    \begin{align}
        e^{i\tilde{E}t} = e^{i\tilde{T_f} C_Z \tilde{T}_f^\dagger t} = \tilde{T_f} (C_{e^{iZt}} \otimes e^{iFt}) \tilde{T}_f^\dagger.
    \end{align}
    All together, we have:
    \begin{align}
        e^{iHt} = M \tilde{W} \tilde{T_f} (C_{e^{iZt}} \otimes e^{iFt}) \tilde{T}_f^\dagger \tilde{W} M^\dagger.
    \end{align}
\end{proof}

The corresponding quantum circuit for Theorem \ref{thm:main} is shown in Figure \ref{fig:onehs}.

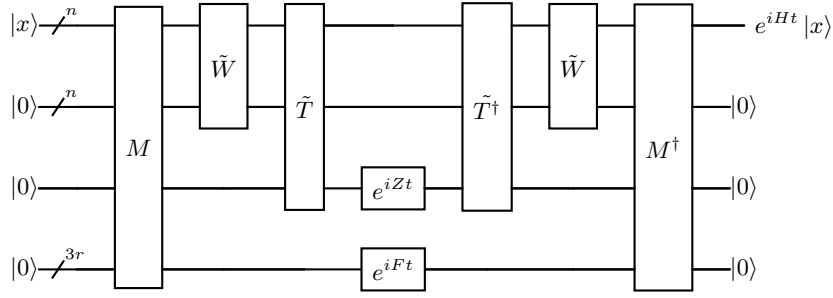
\begin{figure}
    \centering
    \begin{quantikz}[transparent]
        \ket{x} & \qwbundle{n} & \gate[4]{M} & \gate[2]{\tilde{W}} & \gate[3]{\tilde{T}}&\qw            & \gate[3]{\tilde{T^\dagger}}& \gate[2]{\tilde{W}}  & \gate[4]{M^\dagger} & \rstick{$e^{iHt}\ket{x}$} \qw \\
        \ket{0} & \qwbundle{n} & \qw         & \qw                 & \qw                &\qw            & \qw                & \qw                 & \qw & \ket{0} \qw \\
        \ket{0} & \qw          & \qw         & \qw                 & \qw                &\gate{e^{iZt}} & \qw                & \qw                 & \qw & \ket{0} \qw \\
        \ket{0} & \qwbundle{3r}& \qw         & \qw                 & \qw                &\gate{e^{iFt}} & \qw                & \qw                 & \qw & \ket{0} \qw
    \end{quantikz}
    \caption{Quantum circuit for one sparse Hamiltonian simulation.}
    \label{fig:onehs}
\end{figure}

In the quantum circuit \ref{fig:onehs}, the oracle $M$ is taken as assumption and we do not count into resource estimation, $e^{iZt}$ is a single qubit rotation gate. Next we'll discuss the implementation of the operators $\tilde{W}$, $\tilde{T}$, and $e^{iFt}$.

\subsubsection{Implementation of $e^{iFt}$}
The behavior of $e^{iFt}$ is to apply a phase $e^{i \lambda t}$ to the state $\ket{\lambda}$, where $\ket{\lambda}$ is the binary representation of the value $\lambda$. Supposing
$\lambda$ is greater than 0, we have:
\begin{align}
    e^{i\lambda t} \ket{\lambda}
     & = e^{i \sum_{j=1}^n 2^\lambda_j t} \ket{\lambda_1 \cdots \lambda_n}                                                                           \\
     & = e^{i 2^\lambda_1 t} \ket{\lambda_1} \otimes e^{i 2^\lambda_2 t} \ket{\lambda_2} \otimes \cdots \otimes e^{i 2^\lambda_n t} \ket{\lambda_n}.
\end{align}
Thus we can implement $e^{iFt}$ by applying phase gate $e^{i 2^\lambda_i t}$ on each of the qubit $\ket{\lambda_i}$.

In order to allow negative $\lambda$, we have to add another ancilla $\ket{p}$ to indicate the sign, and making the rotation direction of the phase gates controlled by the ancilla. The corresponding circuit is shown in Figure \ref{fig:eift}.

\begin{figure}
    \centering
    \begin{quantikz}
        \ket{p}       & \ctrl[open]{4} & \ctrl[open]{3}  & \ctrl[open]{1}      & \ctrl{4}       & \ctrl{3}        & \ctrl{1}  &                \\
        \ket{a_{n-1}}& \qw            & \qw             & \gate{e^{2^{n-1}it}}& \qw            & \qw             & \gate{e^{2^{n-1}it}}&      \\
        \vdots \\
        \ket{a_1}     & \qw            & \gate{e^{2it}}  & \qw                 & \qw            & \gate{e^{-2it}} & \qw       &                \\
        \ket{a_0}     & \gate{e^{it}}  & \qw             & \qw                 & \gate{e^{-it}} & \qw             & \qw                    &
    \end{quantikz}
    \caption{Quantum circuit for $e^{iFt}$.}
    \label{fig:eift}
\end{figure}
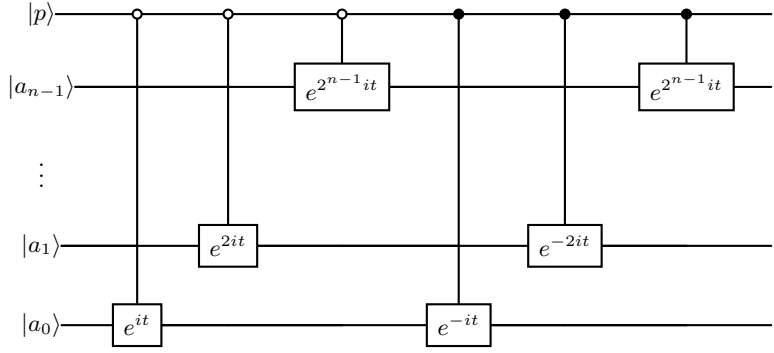

For each controlled phase gate, we can implement it by using the circuit shown in Figure \ref{fig:phase}.

\begin{figure}
    \centering
    \begin{quantikz}
        \ket{p} & \gate{R_z(\frac{\lambda}{2})}&\ctrl{1} & \qw                             & \ctrl{1} & \\
        \ket{a} & \gate{R_z(\frac{\lambda}{2})}&\targ{}  & \gate{R_z(- \frac{\lambda}{2})} & \targ{}  &
    \end{quantikz}
    \caption{Controlled phase gate.}
    \label{fig:phase}
\end{figure}
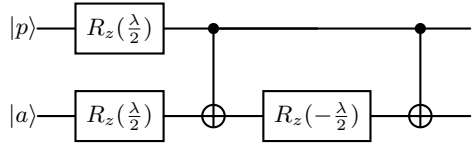

\subsubsection{Implementation of $\tilde{W}$}
Since $\tilde{W} = \bigotimes_{l=1}^{n} W^{l, n+l} $, $\tilde{W}$ is just apply $W$ on $n$ pairs of qubits, the circuit is shown in Figure \ref{fig:tildeW}.
Now for each $W$, since it has the matrix representation \ref{eq:Wmat}, we can implement it using the quantum circuit shown in \ref{fig:W}.

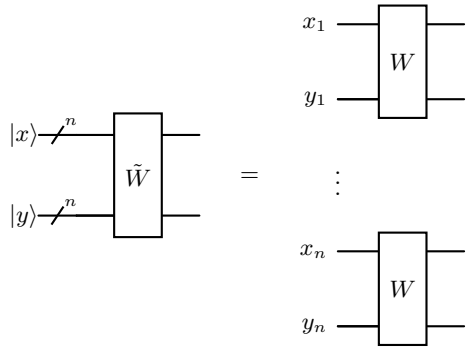
\begin{figure}
    \centering
    \begin{quantikz}
        \ket{x} & \qwbundle{n} & \gate[2]{\tilde{W}} &  \\
        \ket{y} & \qwbundle{n} & \qw         &
    \end{quantikz}
    \quad = \quad
    \begin{quantikz}
        \lstick{$x_1$} & \gate[2]{W} & \\
        \lstick{$y_1$} & \qw         &  \\
        \vdots   \\
        \lstick{$x_n$} & \gate[2]{W} &\\
        \lstick{$y_n$} & \qw     &
    \end{quantikz}
    \caption{Quantum circuit implementing the swap operator $\tilde{W}$.}
    \label{fig:tildeW}
\end{figure}

\begin{figure}
    \centering
    \begin{align}
        \begin{quantikz}
            & \gate[2]{W} & \qw \\
            &             & \qw
        \end{quantikz}
        \quad=\quad
        \begin{quantikz}
            & \targ{}  & \ctrl{1} & \ctrl{1} & \ctrl{1} & \targ{}  & \qw \\
            & \ctrl{-1} & \targ{}  & \gate{H} & \targ{}  & \ctrl{-1} & \qw
        \end{quantikz}
    \end{align}
    \begin{align}
        \quad=\quad
        \begin{quantikz}
            & \targ{} & \ctrl{1}  & \qw      & \qw      & \qw      & \ctrl{1} & \qw            & \qw & \qw & \ctrl{1} & \targ{} & \qw \\
            & \ctrl{-1}  & \targ{} & \gate{S} & \gate{H} & \gate{T} & \targ{}   & \gate{T^\dagger}& \gate{H} & \gate{S^\dagger} &\targ{} &\ctrl{-1} & \qw
        \end{quantikz}
    \end{align}
    \caption{Quantum circuit implementing the gate $W$.}
    \label{fig:W}
\end{figure}

\subsubsection{Implementation of $\tilde{T}$}
The operator $\tilde{T}$ computes the parity of each pair $x_i y_i$ and stores the result in an ancilla qubit. This can be implemented using Toffoli-like gates to compute the parity of each pair $x_i y_i$. The corresponding quantum circuit is shown in Figure \ref{fig:T}.

\begin{figure}
    \centering
    \begin{quantikz}
        \ket{x} & \qwbundle{n} & \gate[3]{\tilde{T_f}} &  \\
        \ket{y} & \qwbundle{n} & \qw         & \\
        \ket{0} & \qw & \qw         &
    \end{quantikz}
    \quad = \quad
    \begin{quantikz}
        \lstick{$x_1$}  & \ctrl{7}       &                 & [0.5cm]   &                  &\\
        \lstick{$y_1$}  & \ctrl[open]{-1}&                 &           &                  &\\
        \lstick{$x_2$}  &                & \ctrl{5}        &           &                  &\\
        \lstick{$y_2$}  &                & \ctrl[open]{-1} &           &                  &\\
        \vdots   \\
        \lstick{$x_n$}  &                &                 &           & \ctrl{2}         &\\
        \lstick{$y_n$}  &                &                 &           & \ctrl[open]{-1}  &\\
        \lstick{$0$}    & \targ{}        & \targ{}         &           & \targ{}          &
    \end{quantikz}
    \caption{Quantum circuit implementing the swap operator $T$.}
    \label{fig:T}
\end{figure}
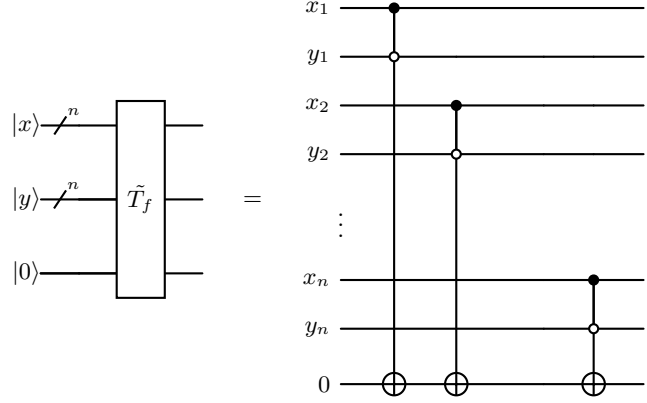

\begin{figure}
    \centering
    \begin{align}
        \begin{quantikz}
            \ket{q_1} & \qw      & \qw      & \qw              & \ctrl{2}  & \qw      & \qw      & \qw             & \ctrl{2}  &\ctrl{1} & \qw            &    \ctrl{1}&\gate{T} &    \\
            \ket{q_2} & \qw      & \ctrl{1} & \qw              & \qw       & \qw      & \ctrl{1} & \gate{T^\dagger}& \qw       &\targ{}  & \gate{T^\dagger}&   \targ{} &\gate{S} &      \\
            \ket{q_3} & \gate{H} & \targ{}  & \gate{T^\dagger} & \targ{}   & \gate{T} & \targ{}  & \gate{T^\dagger}& \targ{}   &\gate{T} & \gate{H} & &&
        \end{quantikz}
    \end{align}
    \caption{Quantum circuit implementing the Toffli-like gate.}
    \label{fig:W_gate}
\end{figure}

\subsection{Logical resource count}
Now we have a concret implementation of one sparse Hamiltonian simulation, we can count the logical resource by adding up the cost of each component. The logical resource
is defined as the number of logical qubits and the number of $\text{Clifford} + T$ gates.

\begin{table}
    \centering
    \begin{tabular}{|c|c|c|c|c|p{8cm}|}
        \hline
        Component      & $T$-gates        & $\operatorname{CNOT}$     & $S$           & $H$           & Explanation                                                                                                                                             \\
        \hline
        $\tilde{W}$    & $4n$             & $10n$      & $4n$          & $4n$          & 2 $\tilde{W}$-gates, each contains $n$ $W$-gates, and each $W$-gate requires 2 $T$-gates, 5 $\operatorname{CNOT}$ gates, 2 $S$-gates and 2 $H$-gates                   \\
        \hline
        $\tilde{T}$    & $14n$            & $12n$      & $2n$          & $4n$          & 2 $\tilde{T}$-gates, each contains $n$ Toffoli-like gates, and each Toffoli-like gate requires 7 $T$-gates, 6 $\operatorname{CNOT}$ gates, 1 $S$-gate and 2 $H$-gates. \\
        \hline
        $e^{iZt}$      & $15$             & $0$        & $3$           & $3$           & Assumes, on average, the rotation gates consume 15 $T$-gates, 3 $S$-gates and $H$-gates.                                                                \\
        \hline
        $e^{iFt}$      & $90r$            & $4r$       & $6r$          & $6r$          & $2r$ controlled rotation gates, each contains 2 $\operatorname{CNOT}$ gate and 3 rotation gates.                                                                       \\
        \hline
        \textbf{Total} & $18n + 90r + 15$ & $22n + 4r$ & $6n + 6r + 3$ & $8n + 6r + 3$ & Summing up all the components.                                                                                                                          \\
        \hline
    \end{tabular}
    \caption{Logical $T$-gate count for each component in the algorithm.}
    \label{tab:Tgatecount}
\end{table}

Thus we have the following theorem:

\begin{theorem}
    For a one-sparse Hermitian matrix $H$ with real valued entries, the gate count of the quantum circuit implementing $e^{iHt}$ is given by:
    \begin{align}
        \begin{aligned}
            T_{\operatorname{count}}    & = 90r + 18n + 15, \\
            \operatorname{CNOT}_{\operatorname{count}} & = 22n + 4r,       \\
            S_{\operatorname{count}}    & = 6n + 6r + 3,    \\
            H_{\operatorname{count}}    & = 8n + 6r + 3.
        \end{aligned}
    \end{align}
    where $r$ is the number of precision bits, $n$ is the number of qubits for input state vector $\ket{x}$.
\end{theorem}

\section{Scaling analysis}\label{sec:scaling}

In this section, we derive the exact scaling of the QLSA algorithm. As mentioned in \ref{sec:QLSABottleneck}, the bottleneck of the QLSA algorithm lies in the number of one-sparse Hamiltonian simulations (nHS1) we have to perform. This is determined by two factors:
\begin{enumerate}
    \item The number of one-sparse Hamiltonian simulations required to implement the s-sparse Hamiltonian simulation (the unitary).
    \item The number of unitaries (the number of clock qubits, nC) in the quantum phase estimation.
\end{enumerate}

It is crucial to recognize that the QLSA algorithm is essentially an iterative method. To suppress the error to a certain level, we need to perform the subroutine for a specific number of iterations. Thus, analyzing the scaling of the QLSA algorithm is equivalent to bounding the error.

The error of the QLSA algorithm is mainly introduced by two parts: the error introduced by the Trotterization and the error introduced by the Quantum Fourier Transform (QFT). Previous error analyses mainly focused on the latter, and a comprehensive error analysis of both parts has not been provided. Here, we will provide a detailed error analysis of the QLSA algorithm.

The main idea to bound the error of both parts simultaneously is to use an exponential-type error to express the error introduced by Trotterization. We define the \textbf{Exponential Type Error} (ETE) as
\begin{equation}
    \text{ETE} := \|\tilde{A} - A\|.
\end{equation}

This error measure quantifies the deviation in the exponent matrix of the Hamiltonian evolution, which can be viewed as the error of the input. Thus, we can separate the error analysis of Trotterization and QFT, and then add the errors together.

Next, we first derive the exponential-type error of Trotterization, followed by the error of QFT. The derivation is heavily inspired by \cite{childs2021theory}. Although the paper does not provide an exact error bound for exponential-type error, it provides the main idea of how to derive the error bound. Therefore, it is relatively straightforward to derive it ourselves. We start with:
\begin{equation}
    f(t) = e^{iH_1 t} e^{iH_2 t} \cdots e^{iH_n t} \label{eq:trot}
\end{equation}
Since f(t) is a unitary operator, we have:
\begin{equation}
    f(t) = e^{iH' t} \label{eq:trot_eqv_rep}
\end{equation}
and we want to bound:
\begin{equation}
    \frac{\|H' - H\|}{\|H\|} \label{eq:re_err}
\end{equation}
Given in the paper, we have:
\begin{equation}
    f(t) = \exp_\mathcal{T} \left( \int_{0}^{t} \mathrm{d}\tau \, (iH + \mathcal{E}(\tau)) \right) \label{eq:exp_err}
\end{equation}
Taking to 1st order in the dysen series expansion of \eqref{eq:exp_err}, we have:
\begin{equation}
    f(t) = I + iHt + \int_{0}^{t} \mathrm{d}\tau \, \mathcal{E}(\tau) + O(t^2) \label{eq:exp_err_1st}
\end{equation}
Also taking 1st order tyler expansion of \eqref{eq:trot_eqv_rep}, we have:
\begin{equation}
    f(t) = I + iH't + O(t^2) \label{eq:trot_eqv_rep_1st}
\end{equation}
Comparing \eqref{eq:exp_err_1st} and \eqref{eq:trot_eqv_rep_1st}, we have:
\begin{equation}
    H' = H + \frac{1}{it} \int_{0}^{t} \mathrm{d}\tau \, \mathcal{E}(\tau)
\end{equation}
Thus:
\begin{equation}
    \|H' - H\| = \frac{1}{t} \left\| \int_{0}^{t} \mathrm{d}\tau \, \mathcal{E}(\tau) \right\|
\end{equation}
In order to get $\mathcal{E}(\tau)$, we differentiate f(t), get:
\begin{align}
    \frac{d f(t)}{d t} & = (i H_1 +  e^{iH_1 t} (i H_2) e^{-iH_1 t} + ... +  e^{iH_1 t} ... e^{iH_{n - 1} t} (iH_n) e^{-iH_{n-1} t} ... e^{-iH_1 t}) f(t) \notag \\
                       & = (\sum_{j = 1}^{n}(\prod_{k=1}^{j-1}e^{iH_k })(iH_j t) (\prod_{k=j-1}^{1}e^{-iH_k t}))f(t)                                      \notag \\
                       & = (\sum_{j = 1}^{n}S_j(t))f(t) \label{eq:trot_frac}
\end{align}
where:
\begin{equation}
    S_j(t) = (\prod_{k=1}^{j -1}e^{iH_k t})(iH_j ) (\prod_{k=j-1}^{1}e^{iH_k t})
\end{equation}

Now we have:
\begin{equation}
    \frac{d f(t)}{d t} = (iH + \mathcal{E}(t))f(t) = (\sum_{j = 1}^{n}S_j(t))f(t)
\end{equation}

\begin{equation}
    \mathcal{E}(t) = \sum_{j = 1}^{n}S_j(t) - iH \label{eq:et}
\end{equation}

Next we will separate $ iHj $ from $S_j(t)$, the way to achieve this is by recursive expansion.
We start with the innermost term. By fundamental theorem of calculus, we have:
\begin{align}
    e^{iH_{j-1} t}(iH_j)e^{-iH_{j-1} t} & = iH_j + \int_0^t d\tau (e^{i H_{j-1} \tau }(i H_{j})e^{-i H_{j-1} \tau})'       \notag             \\
                                        & = iH_j + \int_0^t d\tau e^{i H_{j-1} \tau }[i H_{j-1}, i H_j]e^{-i H_{j-1} \tau} \notag             \\
                                        & = iH_j + \int_0^t d\tau e^{i H_{j-1} \tau }[H_{j}, H_{j-1}]e^{-i H_{j-1} \tau} \label{eq:inner_sep}
\end{align}

For the second layer, we have:

\begin{align}
    e^{iH_{j-2}} e^{iH_{j-1} t}(iH_j)e^{-iH_{j-1} t}e^{-iH_{j-2}} & =  e^{iH_{j-2}t}(iH_j + \int_0^t d\tau e^{i H_{j-1} \tau }[H_{j}, H_{j-1}]e^{-i H_{j-1} \tau})e^{-iH_{j-2}t}      \notag        \\
                                                                  & = e^{iH_{j-2}t} iH_j e^{-iH_{j-2}t}  \notag                                                                                     \\
                                                                  & + (e^{iH_{j-2}t}) \int_0^t d\tau e^{i H_{j-1} \tau }[H_{j}, H_{j-1}]e^{-i H_{j-1} \tau} ( e^{- iH{j-2}t}) \notag                \\
                                                                  & =  iH_j + \int_0^t d\tau e^{i H_{j-2} \tau }[H_{j}, H_{j-2}]e^{-i H_{j-2} \tau} +   \notag                                      \\
                                                                  & +( e^{iH_{j-2}t}) \int_0^t d\tau e^{i H_{j-1} \tau }[H_{j}, H_{j -1}]e^{-i H_{j-1} \tau} (e^{ -iH{j-2}t}) \label{eq:2layer_sep}
\end{align}
Repeat this procedure until all the layers are expanded, we have:
\begin{align}
    S_j(t) & =   (\prod_{k=1}^{j -1}e^{iH_k t})(iH_j ) (\prod_{k=j-1}^{1}e^{iH_k t}) \notag                                                                                                               \\
           & = i H_j + \sum_{k = 1}^{j-1}( \prod_{l=1}^{j-k -1} e^{i H_l t}) \int_0^t d\tau \, e^{i H_{j-k} \tau} [H_j, H_{j-k}] e^{-i H_{j-k} \tau} (\prod_{l=j-k-1}^{1} e^{-i H_j t}) \label{eq:sep_sk}
\end{align}
Then substitute $S_j(t)$ into \eqref{eq:et}, we have:
\begin{align}
    \mathcal{E}(t) & = \sum_{j = 1}^{n}S_j(t) - iH \notag                                                                                                                                                                                     \\
                   & = \sum_{j = 1}^{n} \sum_{k = 1}^{j-1}( \prod_{l=1}^{j-k -1} e^{i H_l t}) \int_0^t d\tau \, e^{i H_{j-k} \tau} [H_j, H_{j-k}] e^{-i H_{j-k} \tau} (\prod_{l=j-k-1}^{1} e^{-i H_j t}) + \sum_{j = 1}^{n} i H_j - iH \notag \\
                   & = \sum_{j = 1}^{n} \sum_{k = 1}^{j-1}( \prod_{l=1}^{j-k -1} e^{i H_l t}) \int_0^t d\tau \, e^{i H_{j-k} \tau} [H_j, H_{j-k}] e^{-i H_{j-k} \tau} (\prod_{l=j-k-1}^{1} e^{-i H_j t}) \label{eq:etau}
\end{align}
Now with $\mathcal{E}(t)$ in hand, we can bound \eqref{eq:re_err} using triangle inequalities:
\begin{align}
    \|H' - H\| & = \frac{1}{t} \left\| \int_{0}^{t} \mathrm{d}\tau \, \mathcal{E}(\tau) \right\| \notag                                                                                                                                                                                  \\
               & \leq \frac{1}{t} \int_{0}^{t}  \mathrm{d}\tau \, \left\| \mathcal{E}(\tau) \right\|  \notag                                                                                                                                                                             \\
               & = \frac{1}{t} \int_{0}^{t} \mathrm{d}t_1 \left\| \sum_{j = 1}^{n} \sum_{k = 1}^{j-1}( \prod_{l=1}^{j-k -1} e^{i H_l t_1}) \int_0^{t_1} \mathrm{d}\tau \, e^{i H_{j-k} \tau} [H_j, H_{j-k}] e^{-i H_{j-k} \tau} (\prod_{l=j-k-1}^{1} e^{-i H_j t_1}) \right\| \notag     \\
               & \leq \frac{1}{t} \int_{0}^{t} \mathrm{d}t_1 \sum_{j = 1}^{n} \sum_{k = 1}^{j-1} \left\| ( \prod_{l=1}^{j-k -1} e^{i H_l t_1}) \int_0^{t_1} \mathrm{d}\tau \, e^{i H_{j-k} \tau} [H_j, H_{j-k}] e^{-i H_{j-k} \tau} (\prod_{l=j-k-1}^{1} e^{-i H_j t_1}) \right\| \notag \\
               & = \frac{1}{t} \int_{0}^{t} \mathrm{d}t_1  \sum_{j = 1}^{n} \sum_{k = 1}^{j-1} \left\|\int_0^{t_1} \mathrm{d}\tau \, e^{i H_{j-k} \tau} [H_j, H_{j-k}] e^{-i H_{j-k} \tau}  \right\| \notag                                                                              \\
               & \leq \frac{1}{t} \int_{0}^{t} \mathrm{d}t_1  \sum_{j = 1}^{n} \sum_{k = 1}^{j-1} \int_0^{t_1} \mathrm{d}\tau \left\| \, e^{i H_{j-k} \tau} [H_j, H_{j-k}] e^{-i H_{j-k} \tau}  \right\| \notag                                                                          \\
               & = \frac{1}{t} \int_{0}^{t} \mathrm{d}t_1  \sum_{j = 1}^{n} \sum_{k = 1}^{j-1} \int_0^{t_1} \mathrm{d}\tau \left\| \,   [H_j, H_{j-k}]   \right\| \notag                                                                                                                 \\
               & \leq \frac{1}{t} \int_{0}^{t} \mathrm{d}t_1  \sum_{j = 1}^{n} \sum_{k = 1}^{j-1} \int_0^{t_1} \mathrm{d}\tau \, 2 \, \max_m \left\| \,   H_m   \right\|^2 \notag                                                                                                         \\
               & = \frac{1}{t} \, \frac{n (n-1)}{2} \, \frac{t^2}{2} \, 2 \, \max_m \left\| \,   H_m   \right\|^2 \notag                                                                                                                                                                  \\
               & = \frac{n(n-1) t }{2}\, \max_m \left\| \,   H_m   \right\|^2 \notag                                                                                                                                                                                                      \\
\end{align}

This bound is tight, as we can take example of $n = 2, t = 0.2, H_1 = [[1,0],[0,-1]]$, and $H_2 = [[0,1],[1,0]]$, the bound and the actual error gives the same value 0.1. However, this bound is significantly above average, the intuition behind this is that
in \eqref{eq:et}, the error term is a sum of all the commutators, and the commutators can
cancel each other, but the bound adds up the norm of all the commutators.

Intuitively this is similar to bounding the sum of random numbers, the bound is the number of random numbers times the maximum of the random numbers.

So the better strategy here is to use probabilistic bound instead, where we tries to
bound the error with failure rate below certain threshold.

\begin{figure}
    \centering
    \includegraphics[width=0.8\linewidth]{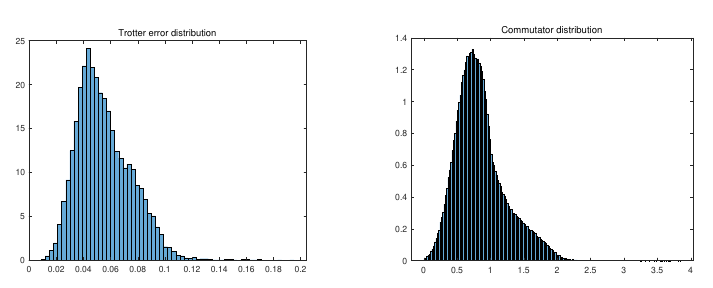}
    \caption{Trotter error and commutator distribution.}
    \label{fig:tcdistr}
\end{figure}

From the numerical result, we can see that the error does not depend on n. Thus, we can simply do the numerical experiment on small matrix size, and take the upper bound around the edge of distribution, we have:
\begin{equation}
    \epsilon = \|\frac{t}{r}\|
\end{equation}
    
Where $t$ is the time step and $r$ is the Trotter repetition.
Now we have the probabilistic bound on the matrix norm, we can transform it into error on $x$ by applying the following formula:

\begin{equation}
    \frac{\|\Delta {x}\|}{\|{x}\|} \leq \frac{\kappa \frac{\|\Delta A\|}{\|A\|}}{1 - \kappa \frac{\|\Delta A\|}{\|A\|}} \label{eq:err_x}
\end{equation}

Since $\|A\| = 1$, and $\|x\| = 1$, we have:
\begin{equation}
\begin{split}
    \|\Delta {x}\| \leq \frac{\kappa \|\Delta A\|}{1 - \kappa \| \Delta A\|} \leq \frac{\kappa t}{r - \kappa t} \label{eq:deltax}
\end{split}
\end{equation}


\begin{equation}
     \| \Delta x_1\| \leq \sqrt{\frac{20}{3}} \pi \frac{k}{t_0}
\end{equation}

Under the constraint:

\begin{equation}
    t_0 \leq 2 \pi T \Rightarrow \frac{1}{T} \leq 1 - \frac{2 \pi \kappa}{t0}
\end{equation}

Here, $T$ represents the number of repetitive unitaries, which is 2 to the power of number of clock qubits, and $t_0$ is the maximum evolution time, given by $t T$. These constraints ensure that the eigenvalues fall within the range of the QPE estimation.
For simplicity, let's assume $t_0 = \pi T$ to meet the first constraint, and later we'll see that the second constraint is satisfied automatically. Then the error in this step is given by:

\begin{equation}
    \| \Delta x_2\| \leq \sqrt{\frac{20}{3}} \frac{k}{T}
\end{equation}

Now let's combine the error in both steps, by using the triangular inequality:

\begin{equation}
    \| \Delta x\| = \| \Delta x_1 + \Delta x_2\| \leq \| \Delta x_1\| + \| \Delta x_2\|
\end{equation}

For simplicity we let $\| \Delta x_1 \|\leq \frac{\epsilon}{2}$ and $\| \Delta x_2 \|\leq \frac{\epsilon}{2}$, then we have:

\begin{equation}
    r \geq \left(\frac{2}{\epsilon} + 1\right) \pi \kappa \approx \frac{2 \pi \kappa}{\epsilon}
\end{equation}

\begin{equation}
    T \geq \sqrt{\frac{80}{3}} \frac{\kappa}{\epsilon}
\end{equation}

Finally, we combine those two to get the number of HS1 operations needed:

\begin{equation}
    n_{\text{HS1}} = T r s \geq \frac{\sqrt{\frac{320}{3}} \pi \kappa^2 s}{\epsilon^2}
\end{equation}

Where we assume that there are $s$ number of HS1 per HS, this assumption requires that the oracle needs to implement the seperation of an $s$-sparse matrix into s number of one sparse matrices, this does not hold in general, but for applications with certain structures this can be done.

\subsection{Logical resource scaling}
Before we derived the scaling of the logical resource of each HS1 operation, and in this section we derived the scaling of number of HS1 operations. By multiply both scaling together, we have the following theorem:

\begin{theorem}
    The scaling of the logical gates is given by:
    \begin{align}
        n_{\operatorname{gates}} = \frac{\sqrt{\frac{320}{3}} \pi \kappa^2 s}{\epsilon^2} \times \operatorname{HS1}_{\operatorname{gates}}
    \end{align}
\end{theorem}

\begin{proof}
    The proof is straightforward, we just need to multiply the scaling of the logical resource of each HS1 operation with the scaling of the number of HS1 operations.
\end{proof}

Finally, compared to \cite{harrow2009quantum}, the reduction in $ s $-dependence is achieved by assuming that the matrix is perfectly decomposed into exactly $ s $ sub matrices at the oracle level. Although this assumption does not hold for general sparse matrices, many QLSA applications involve fixed matrix structures, which allow for efficient oracle implementation that satisfies this assumption, for example, in \cite{PLiu_2021} and \cite{Liu_2024}. If we remove this assumption, we'll have to additionally implement the decomposition procedure, one approach is to utilize graph coloring algorithm, which decompose the matrix into $6s^2$ one sparse matrices, thus increase a prefactor of $6$ and square the dependence on sparsity. Furthermore, the graph coloring procedure is being called in each HS1, and according to our implementation, the graph coloring procedure involve around $200\times \log N$ $T$-gates, this cost is about $10$ times the other cost in HS1 ($18\log N + 90r + 15$), thus introducing another prefactor of $10$, so overall, the cost for graph coloring based general QLSA will have a scaling of $60sT$. 

Assuming $(\epsilon = 0.1, \kappa = s = \log_2{N})$, we have the scaling of the different gates in Figure \ref{fig:gate_scaling}. 

\begin{figure}
    \centering
    \includegraphics[width=0.7\textwidth]{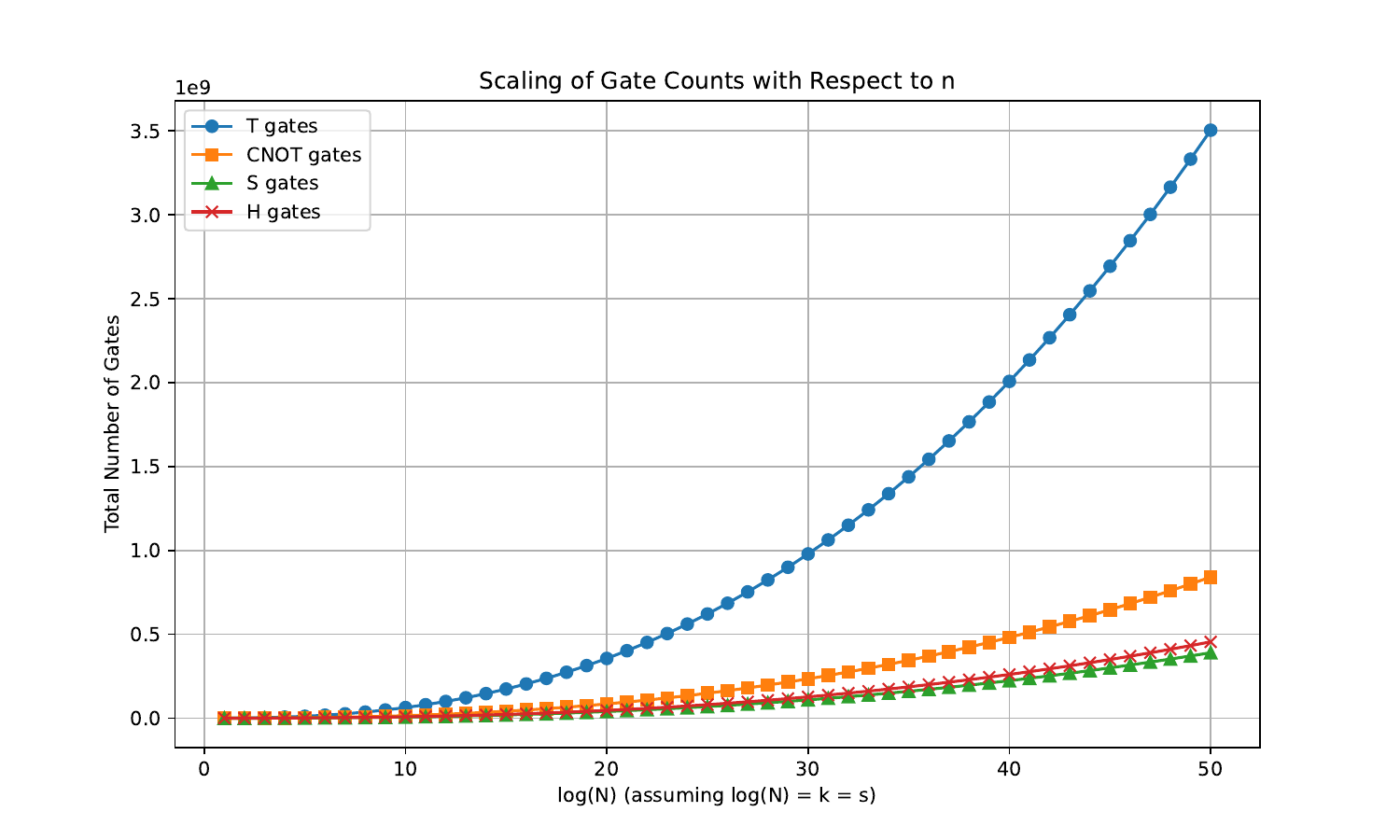}
    \caption{Scaling of the logical gates with respect to the problem size.}
    \label{fig:gate_scaling}
\end{figure}

\section{Surface code protocol}

\subsection{Introduction to surface code}

To reliably store and process quantum information over extended periods, active error correction is essential. This is achieved by encoding multiple physical qubits into logical qubits using quantum error-correcting codes \cite{Preskill_1998,RevModPhys.87.307,Campbell_2017}. Among these, the surface code is not only the most widely used but also one of the most crucial due to its compatibility with the locality constraints of practical quantum hardware, such as superconducting qubits, which only support two-dimensional local operations \cite{Kitaev_2003, PhysRevA.86.032324}.

However, despite its importance and widespread adoption, the surface code introduces significant computational overhead. The replacement of physical qubits with logical qubits drastically increases space requirements, while the restriction to two-dimensional local operations imposes additional time costs. Arbitrary quantum gates may require multiple time steps instead of executing in a single step, making the choice of surface code schemes—which define parameters such as code distance, logical gate protocols, and resource allocation—critical in determining the actual overhead.

To accurately assess the resource demands of a quantum algorithm, it is essential to analyze the specific surface code scheme under which it is executed. By optimizing these schemes, we can minimize space-time overhead and better evaluate the feasibility of quantum computations within a surface-code-based architecture.

\subsection{Key concepts}

\subsubsection{Patch}
A \emph{patch} is a two-dimensional regular lattice of entangled physical qubits that serves as the substrate on which logical qubits are defined. Physical qubits within a patch are categorized into two types:

\begin{itemize}
    \item \textbf{Data qubits}: These qubits store quantum information and are measured less frequently, primarily during computational operations.
    \item \textbf{Syndrome qubits}: These qubits interact repeatedly with neighboring data qubits and are frequently measured to detect the presence of errors.
\end{itemize}

Logical qubits within a patch can perform logical operations (or gates) using a technique known as \emph{lattice surgery}, which is beyond the scope of this discussion. The code distance, denoted as $d$, determines the error-correcting capability of the patch. A patch of code distance $d$ consists of $d^2$ physical qubits, with larger values of $d$ yielding higher fidelity computations.

\subsubsection{Blocks}
Blocks are functional units that organize multiple patches according to specific rules. Blocks are categorized as follows:

\begin{itemize}
    \item \textbf{Data blocks}: These accommodate logical data qubits and execute logical operations, including logical gates.
    \item \textbf{Distillation blocks}: These are responsible for generating \emph{magic states}, which are necessary for executing certain non-Clifford gates.
\end{itemize}

For Clifford gates, computations can be efficiently performed within the data block. However, executing a $T$-gate requires magic state resources, making it significantly more challenging. The distillation block produces magic states, which are then consumed by the data block to enable $T$-gate operations. Since both processes rely on complex and time-consuming protocols, $T$-gates become the primary bottleneck of quantum computation.

\subsubsection{Protocols}
Protocols define how patches within data blocks and distillation blocks are organized and how magic states are produced and consumed. Protocol selection is crucial in optimizing the space-time trade-off of quantum computations. 

\begin{itemize}
    \item \textbf{Distillation protocols}: By allocating more patches to a distillation block, magic states can be generated more quickly or with higher fidelity.
    \item \textbf{Data block protocols}: Increasing the number of patches in a data block enables faster consumption of magic states, enhancing computational speed.
\end{itemize}

\begin{table}\label{tab:distillation_proto}
    \centering
    \renewcommand{\arraystretch}{1.3}  
    \setlength{\tabcolsep}{10pt}       
    \caption{Distillation Protocol Parameters}
    \label{tab:distillation}
    \begin{tabular}{|l|c|c|c|}
        \hline
        \textbf{Protocol} & \textbf{Number of Magic Tiles} & \textbf{Production Time} & \textbf{Error Rate Coefficient } \\
        \hline
        15-1              & 11                             & 11                       & $35p^3$                               \\
        \hline
        116-12            & 44                             & 9.27                     & $4.125p^4$                            \\
        \hline
        225-1             & 176                            & 5.5                      & $1.5p^7$                              \\
        \hline
    \end{tabular}
\end{table}

\begin{table}
    \centering
    \renewcommand{\arraystretch}{1.3}  
    \setlength{\tabcolsep}{10pt}       
    \caption{Data Protocol Parameters}
    \label{tab:data}
    \begin{tabular}{|l|c|c|}
        \hline
        \textbf{Protocol} & \textbf{Magic Consumption Time} & \textbf{Number of Tiles per $n$ Qubits} \\
        \hline
        Compact           & 9                               & $1.5n + 3$                              \\
        \hline
        Intermediate      & 5                               & $2n + 4$                                \\
        \hline
        Fast              & 1                               & $2n + \sqrt{8n} + 1$                    \\
        \hline
    \end{tabular}
\end{table}

A comprehensive set of protocols is described in \cite{litinski2019game}. The parameters of these protocols are summarized in Table \ref{tab:distillation} and \ref{tab:data}. Later in this work, we optimize the space-time cost of an algorithm with a given $T$-gate count and physical parameters by selecting appropriate protocols.

\subsubsection{Surface code scheme}
A \emph{surface code scheme} provides a complete description of a surface code implementation. It is defined by four key parameters:

\begin{enumerate}
    \item The code distance of the patches.
    \item The data block protocol.
    \item The distillation protocol.
    \item The number of distillation blocks.
\end{enumerate}

These parameters collectively determine the efficiency and reliability of the quantum computation. By carefully selecting and optimizing these elements, it is possible to achieve an optimal balance between resource utilization and computational fidelity.

\subsection{Surface-code-based physical resource estimator}
The physical resource estimator is a function that takes the logical resource requirements and quantum computer hardware parameters as inputs and outputs the surface code scheme with the lowest space-time cost. It also provides the corresponding runtime and the number of physical qubits required.

Fundamentally, finding the optimal surface code scheme is a constrained optimization problem. The primary constraint is that both the magic state and the logical qubits must achieve a certain level of fidelity to ensure the final error rate remains below a predefined threshold of 0.01. The objective of the optimization is to minimize the space-time cost, which is defined as the product of runtime and the number of physical qubits. The optimization parameters consist of four key aspects of the surface code scheme: the distillation protocol, the number of distillation blocks, the data block protocol, and the code distance.

We have implemented the optimization procedure outlined in \cite{litinski2019game} to determine the most efficient surface code scheme. The corresponding pseudocode can be found in \ref{fig:psud}.

\begin{figure}
    \centering
    \includegraphics[width=0.75\linewidth]{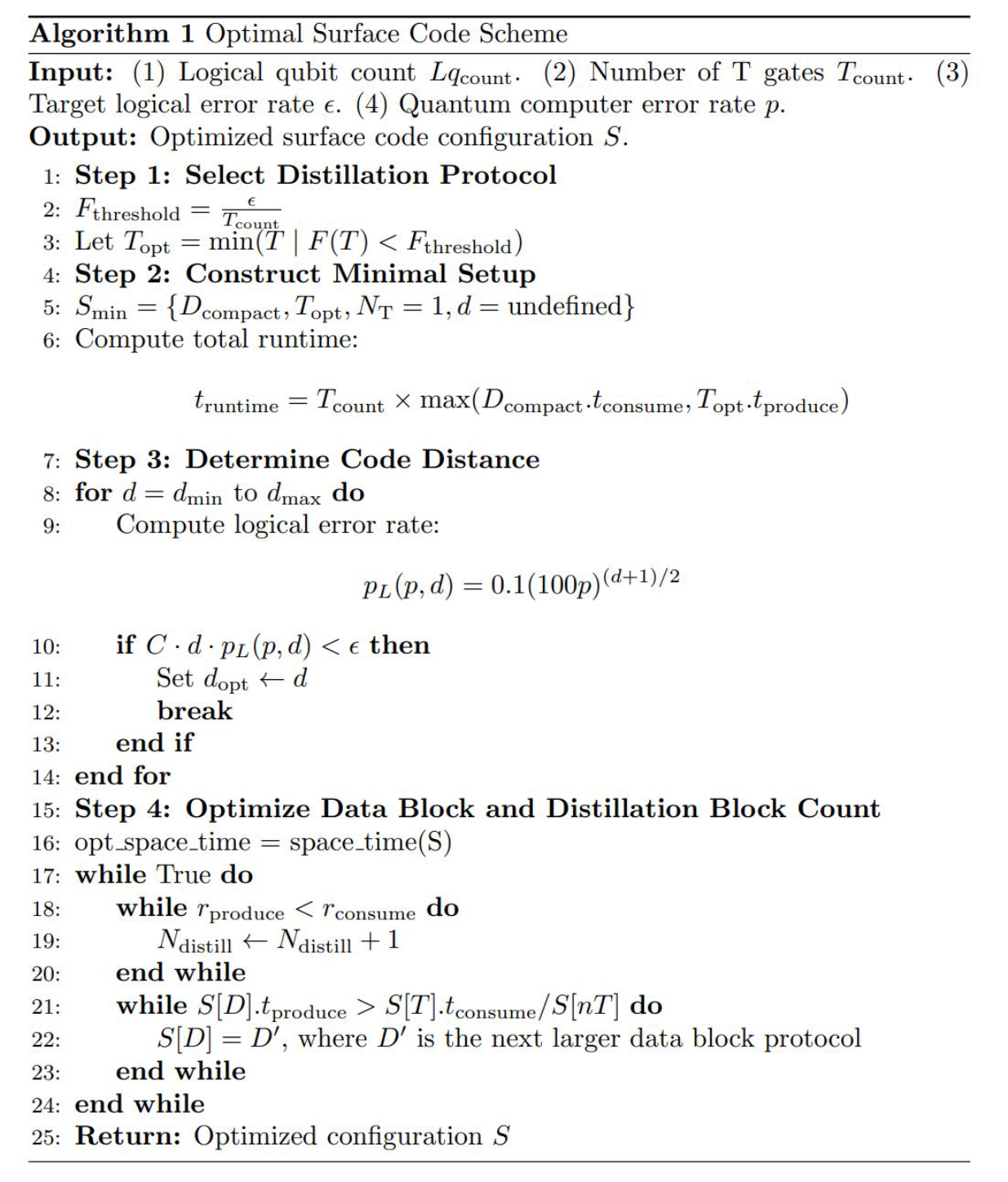}
    \caption{Optimizer Pseudo Code}
    \label{fig:psud}
\end{figure}

The procedure consists of four steps, which we describe as follows:

\subsubsection{Step 1: Determine the distillation protocol}
A more sophisticated distillation protocol produces magic states with higher fidelity but at the cost of using more patches (and thus more physical qubits). The parameters of several distillation protocols are listed in \cite{litinski2019game}. In this step, our goal is to select the most cost-effective distillation protocol that achieves sufficient fidelity to keep the overall computation error rate below 0.01. 

The fidelity threshold is defined as the error rate divided by the total number of $T$-gates. We simply choose the least expensive distillation protocol that surpasses this fidelity threshold.

\subsubsection{Step 2: Construct a minimal setup}
In this step, we construct a surface code scheme that occupies the smallest possible space. While this minimal setup is not optimal, it serves as a starting point to determine the code distance. Subsequent optimizations of the data block protocol will build upon this minimal setup.

The minimal setup consists of a compact data block combined with a single distillation block determined in Step 1. The total runtime of this setup is given by the number of $T$-gates multiplied by the clock cycle required to execute each $T$-gate. The latter is determined by the slower operation between magic state production and consumption.

For example, referring to \ref{tab:data}, the compact data block consumes a magic state every 9 clock cycles, while the 15-to-1 distillation block occupies 11 tiles and outputs a magic state every 11 clock cycles. Since the slower operation dictates the runtime, the effective clock cycle for $T$-gate execution is 11 cycles. If the total $T$-gate count is $10^8$, then the algorithm completes in $11 \times 10^8$ time steps.

\subsubsection{Step 3: Determine the code distance}
A higher code distance reduces the logical qubit error rate but requires more physical qubits. Denoting the code distance by $d$, the number of physical qubits per patch is $d^2$. The logical error rate per logical qubit per code cycle can be approximated as [12]:

\begin{equation}
	p_L(p,d) = 0.1(100p)^{(d+1)/2} \, .
\end{equation}

The code distance must be sufficiently large to suppress the logical error rate such that the total logical error probability for the entire algorithm remains below 0.01. This requirement translates into the following condition:

\begin{equation}
	164 \times 11 \times 10^8 \times d \times p_L(10^{-4},d) < 0.01 \, .
\end{equation}

\subsubsection{Step 4: Optimize the data block protocol and distillation block count}

Starting from the minimal setup, we identify the performance bottleneck, which typically lies in the speed of magic state production. To accelerate $T$-gate execution, we increase the number of distillation blocks until the magic state production rate surpasses the consumption rate. At this point, we switch to a larger data block protocol to further enhance magic state consumption.

In principle, repeating this process eventually yields a time-optimal surface code scheme, though not necessarily a space-time optimal one. However, for our limited set of distillation and data block protocols, we observe that the time-optimal scheme coincides with the space-time optimal scheme. This suggests that, in the absence of physical qubit constraints, the optimal surface code scheme always employs a \emph{fast block} configuration, with a number of distillation blocks precisely matching the magic state consumption rate.

Thus, in the setting of infinite physicla qubits, what our optimizer does is simply decide the distillation protocol, and code distance, then the data protocol must be fast block, and the number of distillation block must be the smallest number such that the speed of magic state production matches the magic state consumption speed of fast block. 

However, more intriguing patterns emerge when the number of physical qubits is constrained. The optimized surface code scheme, as shown in Figure~\ref{fig:opt-proto}, illustrates how the optimal strategy evolves with increasing matrix size. When the matrix size is small, the surface code prioritizes a time-optimal scheme. As the matrix size grows (the sparsity also grows with $\log N$), the scheme gradually reduces the distillation block count, reaching a minimum of three distillation blocks. Beyond this point, further increases in matrix size lead to a transition in data block types, shifting from a standard block to an intermediate block, and eventually to a compact block. Notably, as the data block type changes, the runtime increases significantly.

\begin{figure}
    \centering
    \includegraphics[width=1.0\linewidth]{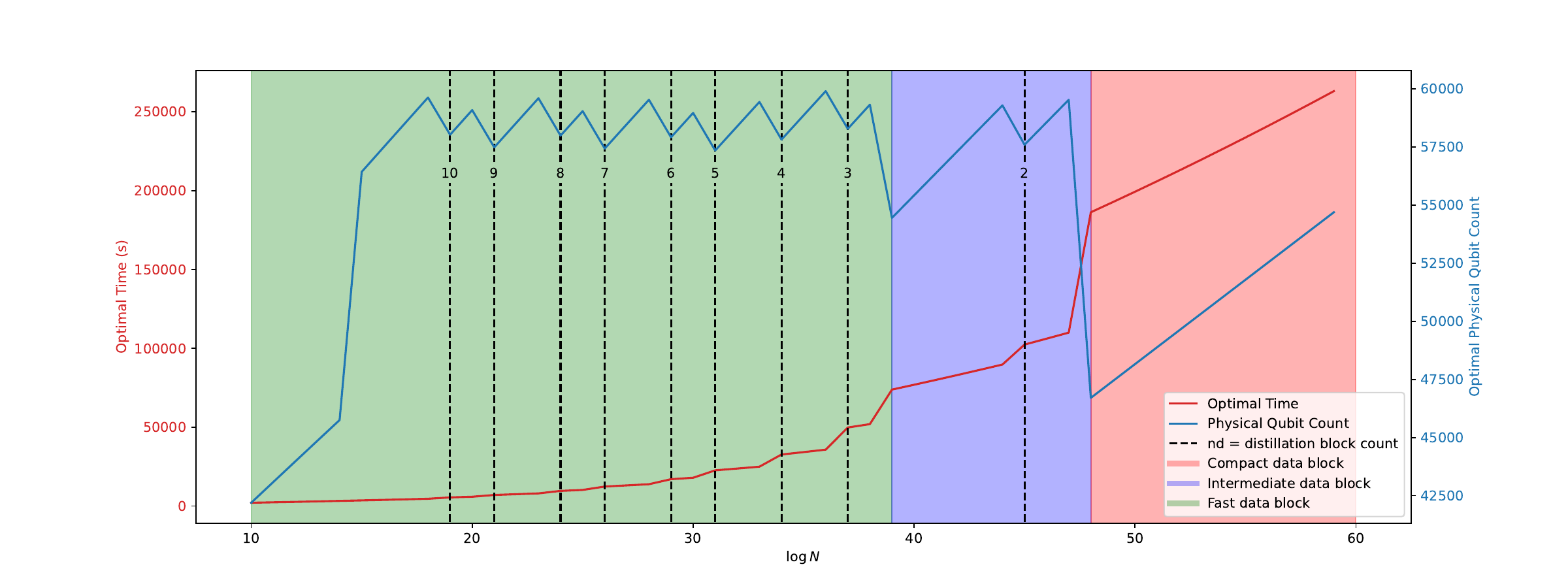}
    \caption{\textbf{Optimizer Result.} This figure shows the optimization results for different protocols, where we used the same setting as \ref{fig:big_figure} except that the physical qubits is restricted to 60000. The red and blue lines represent the optimal time and physical qubit count, respectively. The black dashed lines indicate changes in the distillation block count. The colored background represents different data block protocols.}
    \label{fig:opt-proto}
\end{figure}

\section{The classical counterparts} \label{sec:clascaling}

\subsection{Conjugate Gradient Method}

In this section we introduce the runtime resource estimation of the Conjugate Gradient method. 

\begin{definition}[LSP]\label{def: LSP}
The Linear System Problem: Find a vector $x\in \mathbb{R}^{n}$ such that it satisfies equation $Ax=b$ with coefficient matrix $A\in \mathbb{R}^{m \times n}$ and right-hand side (RHS) vector $b\in \mathbb{R}^n$.
\end{definition}

A basic approach for solving an LSP is Gaussian elimination, or LU factorization,  with ${O}(N^3)$ arithmetic operations. 
If $A$ is a square symmetric positive semi definite (PSD) matrix, we can also apply Cholesky factorization with ${O}(N^3)$ arithmetic operations. 
The best complexity for an iterative algorithm with respect to $N$ is ${O}(Ns\sqrt{\kappa}\log(1/\epsilon))$ arithmetic operations for the Conjugate Gradient (CG) method solving systems with symmetric PSD matrices, where $s$ is the maximum number of non-zero elements in any row or column of $A$, $\kappa$ is the condition number of $A$, and $\epsilon$ is the error allowed. If matrix $A$ is just symmetric, one can use the Lanczos algorithm with higher complexity.
For an LSP with a general square matrix $A$, the best iterative method is the GMRES algorithm, which has ${O}(n^3)$ worst-case complexity \cite{iterativelsp}. 
For problems in the form of $E^TE x = E^T \psi$, known as normal equations, one can use a version of CG methods with complexity ${O}(nd\kappa_{E}\log(1/\epsilon))$, where $\kappa_E$ is the condition number of matrix $E$ \cite{iterativelsp}. 
For a linear system in general form with a non-PSD non-symmetric matrix, one can use the reformulation $A^TA x =A^T b$ and use a CG method to solve it. 
Although CG methods for this reformulation have better worst-case complexity than GMRES for the original system, practically GMRES has better performance, especially for large sparse systems with a large condition number \cite{iterativelsp}.

Here, we want to use CGNE (Algorithm 8.5 of \cite{iterativelsp}) as indicated in Algorithm~\ref{fig:cgne}.

\begin{figure}
    \centering
    \includegraphics[width=0.75\linewidth]{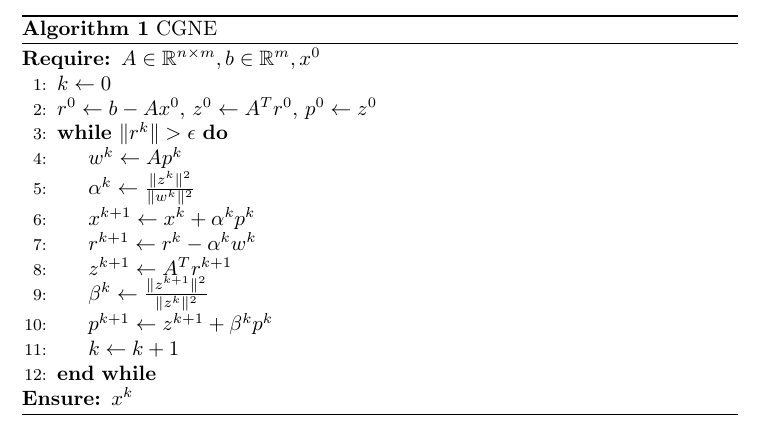}
    \caption{CGNE Algorithm}
    \label{fig:cgne}
\end{figure}

As we can see, there is no matrix-matrix product in the CG algorithm. At each iteration, we need to perform two mat-vec products. For each sparse mat-vec product, we require $2Ns$ FLOPs. Additionally, we need three times vector summation with $2N$ FLOPs. Also, $8N$ FLOPs needed for calculating $\alpha$ and $\beta$. Total FLOPs needed in each iteration is 
$4Ns+6N+8N.$
Thus the dominant cost is for mat-vec products.
\begin{theorem}
    Starting from $x^0=0$, the algorithm in Figure \ref{fig:cgne} reaches to $\frac{\|x^{k}-x\|}{\|x\|}\leq \epsilon$ after 
    $k\geq \frac{1}{2}\kappa \log(\frac{2}{\epsilon}).$
\end{theorem}
Proof can be found in \cite{Shewchuk1994}. 

For a $s$-sparse Hermitian matrix, the total number of the FLOPs are 

 \begin{align}(\frac{1}{2}\kappa \log(\frac{2}{\epsilon}))\times (4Ns+6N+8N)\end{align}

\subsection{Cholesky Decomposition Method}

Cholesky decomposition, combined with subsequent Gaussian elimination, can be efficiently parallelized with minimal overhead, particularly for large matrices. This allows us to maintain intermediate matrices in a sparse form, enabling storage within GPU memory.

The algorithm based on Cholesky decomposition consists of the following steps:

\begin{enumerate}
    \item Given a permutation of rows and columns of the linear equation matrix $P$ and vector $b$, optimize the process for solving $Px = b$.
    \item Construct the necessary data structures to compute and store $P$.
    \item Perform Cholesky decomposition: $P = L L^\top$, where $L$ is a sparse lower triangular matrix.
    \item Solve the triangular systems: first, solve $L y = b$ for $y$, then solve $L^\top x = y$ for $x$.
\end{enumerate}

For comparison with quantum algorithm performance evaluations, we exclude GPU data loading and unloading time from our analysis.

Additionally, we conclude that each step of the algorithm is fully parallelizable, with the exception of synchronization overhead, which remains negligible.

We calculate the number of floating-point-operations (FLOPs) for each step, where:
\begin{itemize}
    \item $N$ represents the size of matrix $P$ and vector $b$.
    \item $s$ is the maximum number of non-zero elements per row (or column), assuming the matrix is symmetric.
\end{itemize}

\textbf{Step 1: Permutation of Rows and Columns}

Exploiting matrix symmetry, we traverse each row and swap matrix elements as necessary, leading to:
\begin{equation}
    \text{FLOPs} = N \cdot s \cdot (s + 1)
\end{equation}

\textbf{Step 2: Data Structure Construction}

Assuming the matrix is already stored in an optimal format, the primary operation is generating the index structures, which takes:
\begin{equation}
    \text{FLOPs} = N \cdot s
\end{equation}

\textbf{Step 3: Cholesky Decomposition}

For each row, we perform factorization computations, resulting in:
\begin{equation}
    \text{FLOPs} = N \cdot (1 + s + 2s^2)
\end{equation}

\textbf{Step 4: Triangular Solve}

This involves solving two triangular systems, first for $y$ and then for $x$:
\begin{equation}
    \text{FLOPs} = N \cdot (1 + s) \cdot 4
\end{equation}

\textbf{Total FLOPs}

Summing all the FLOPs from the above steps:
\begin{align}\label{eq:CDFLOPs}
    \text{Total FLOPs} &= N \cdot s \cdot (s+1) + N \cdot s + N \cdot (1 + s + 2s^2) + N \cdot (1 + s) \cdot 4 \\
    &= N \cdot (3s^2 + 7s + 5)
\end{align}

Since the CD method is fully parallelizable, we can estimate the runtime by simply divide the FLOPS by the supercomputer's execution speed. According to \cite{top500_2024}, some of the state-of-the-art supercomputer parameters are summerized in \ref{tab:supercomputer_FLOPs}. 

\begin{table}[h]
    \centering
    \begin{tabular}{|l|l|}
        \hline
        \textbf{Name} & \textbf{FLOP/s} \\
        \hline
        Aurora & $1.012 \times 10^{18}$ \\
        El Capitan FP64 & $2.726 \times 10^{18}$ \\
        El Capitan FP32 & $5.453 \times 10^{18}$ \\
        El Capitan FP16 & $4.361 \times 10^{19}$ \\
        El Capitan FP8 & $8.721 \times 10^{19}$ \\
        Frontier FP64/32 & $1.810 \times 10^{18}$ \\
        \hline
    \end{tabular}
    \caption{Supercomputer FLOP/s Comparison}
    \label{tab:supercomputer_FLOPs}
\end{table}

Using these parameters—assuming $ k = 10 $, $ s = \log N $, $ \epsilon = 0.01 $, and a quantum computer clock cycle of 10 ns—along with a fast block implementation where a logical $T$-operation is executed per clock cycle, we can compute the runtime of the CD method on a supercomputer and compare it with the Quantum Linear Systems Algorithm (QLSA). The results of this comparison are presented in Figure \ref{fig:CDruntime}.  

\begin{figure}
    \centering
    \includegraphics[width=0.5\linewidth]{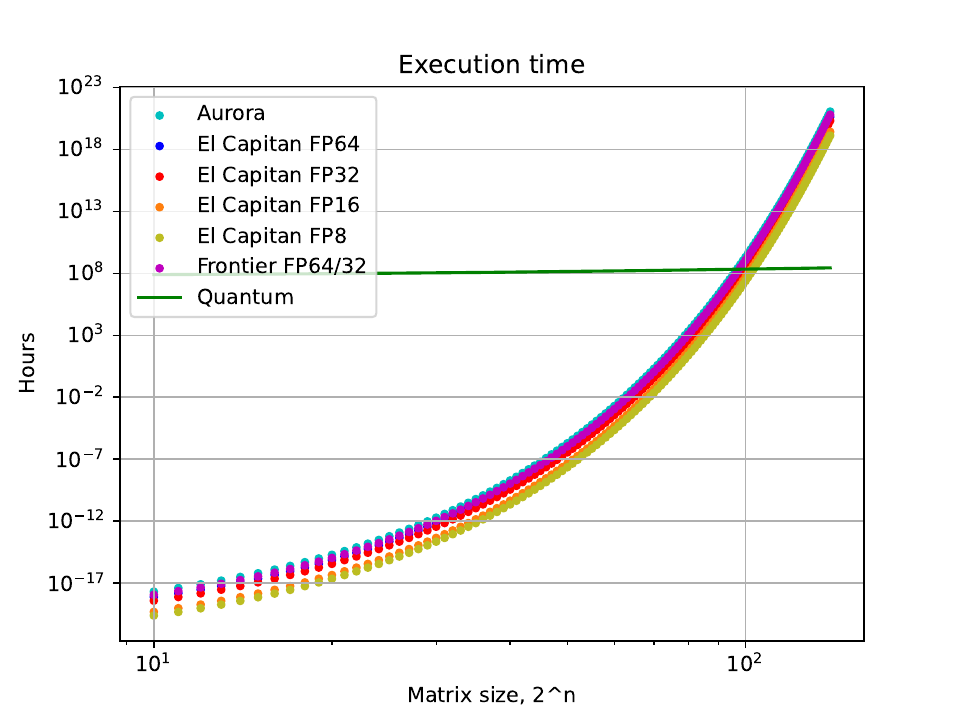}
    \caption{Runtime comparison of the CD method across different supercomputers.}
    \label{fig:CDruntime}
\end{figure}  

Our analysis indicates that the CD method exhibits a time complexity of $ 3Ns^2 $, which, in contrast to the CG method (see Statement \ref{statement:classical}), has a less favorable dependence on $ s $. However, a key advantage of this algorithm is its high degree of parallelizability, enabling efficient execution on supercomputers. Leveraging this capability significantly enhances its performance in terms of runtime. Consequently, the onset of potential quantum advantage may be delayed, as illustrated in Figure \ref{fig:CDruntime}, where it is projected to emerge at a problem size of $ 2^{100} \approx 10^{30} $. Nevertheless, since the runtime of QLSA doesn't scaling much with matrix size, the runtime required to reach the potential quantum advantage threshold remains approximately 200 hours.

\section{Quantum energy analysis}

The energy consumption of a quantum computer can be estimated as the product of three factors: the scale of the quantum computer (i.e., the number of physical qubits), the energy efficiency of the quantum computer (i.e., power consumption per qubit), and the operating time. This relationship is expressed as:
\begin{equation}
    E = n_q \times P_q \times T, \label{eq:total_energy}
\end{equation}
where:
\begin{itemize}
    \item $E$ is the total energy consumption,
    \item $n_q$ is the number of physical qubits,
    \item $P_q$ is the power consumption per physical qubit, and
    \item $T$ is the total operating time.
\end{itemize}

The values of $n_q$ and $T$ are determined from prior analysis, leaving $P_q$ as the main parameter to be explained in detail.

\subsubsection{Power consumption per qubit}

As derived in \cite{martin2021energyusequantumdata}, the power consumption per physical qubit, $P_q$, can be expressed as:
\begin{equation}
    P_q = q \left[ 1 + \phi \left( \frac{1 + \beta n_p^{-1/3}}{\eta_c \text{COP}(T_c)|_c} \right) + \frac{1 - \phi}{\text{FOM}(T_o)} \right], \label{eq:power_consumption}
\end{equation}
where:
\begin{itemize}
    \item $q$ is the computational power per physical qubit,
    \item $\phi$ represents the fraction of power used for direct cooling,
    \item $\beta$ is the external heat conduction ratio,
    \item $n_p$ is the number of physical qubits per logical qubit,
    \item $\eta_c$ is the efficiency of the cooling system,
    \item $\text{COP}(T_c)|_c$ is the coefficient of performance (COP) of the cooling system at the chip temperature $T_c$, and
    \item $\text{FOM}(T_o)$ is the figure of merit for the secondary cooling system.
\end{itemize}

The second term in Eq.~\eqref{eq:power_consumption} accounts for direct heat dissipation, while the third term represents secondary cooling power. Given the early development stage of quantum computers, accurate hardware parameter values for future large-scale fault-tolerant quantum computers remain uncertain, and parameter estimates can vary widely.

\subsubsection{Simplified estimation of $P_q$}

For a rough estimation of $P_q$ based on current quantum computer parameters, we make the following assumptions:
\begin{itemize}
    \item The direct heat dissipation power dominates over computational power and secondary cooling power, allowing the omission of the first and third terms in Eq.~\eqref{eq:power_consumption}.
    \item The volume of the quantum computer is proportional to the number of physical qubits.
\end{itemize}

Under these assumptions, $P_q$ simplifies to:
\begin{equation}
    P_q = \frac{\tilde{q}}{\eta_c \text{COP}(T_c)|_c}, \label{eq:simplified_power}
\end{equation}
where $\tilde{q}$ is a parameter to be determined.

\subsubsection{Determination of $\tilde{q}$}

Based on data from \cite{parker2023estimating}, IBM's Quantum System Two dilution refrigerator can house 4,158 qubits while consuming 26 kW of power. This yields a per-qubit power consumption of approximately $6.25$ W:
\begin{equation}
    P_q = \frac{26,000 \text{ W}}{4,158} \approx 6.25 \text{ W}.
\end{equation}

For superconducting qubits used by IBM, the coefficient of performance (COP) of the cooling system is approximately $10^{-5}$. Substituting into Eq.~\eqref{eq:simplified_power}, we obtain:
\begin{equation}
    \tilde{q} = 6.25 \times 10^5.
\end{equation}

\subsubsection{Range of $P_q$ for different architectures}

Assuming $\tilde{q}$ remains constant across different architectures, $P_q$ depends primarily on the COP. For a COP ranging from $10^{-5}$ to $10^{-2}$, the power consumption per qubit varies between $6.25$ W and $0.0625$ W.

\subsubsection{Conclusion}

This analysis provides an approximate estimation of the energy consumption for quantum computers, highlighting the dependence of power consumption on cooling efficiency and quantum computer architecture. Future work will refine these estimates as hardware parameters become more concrete.

\section{Classical energy analysis}
The estimation of classical energy consumption follows the formula:

\begin{equation}\label{eq:cenergy}
    E_c = t P_c = \frac{op}{f_c} P_c = op \frac{P_c}{f_c}
\end{equation}

where $E_c$ represents the total energy consumption, $t$ denotes the runtime, $op$ is the total number of clock cycles, $f_c$ stands for the clock frequency, and $P_c$ is the power consumption. Consequently, determining the energy efficiency of CPUs, expressed as $\frac{P_c}{f_c}$, is crucial for performance evaluation.

Most mainstream desktop CPUs exhibit power consumption ranging from 50 Watts to 400 Watts, with clock frequencies spanning from 1 GHz to 5 GHz. Given that power and frequency are positively correlated~\cite{cpupower}, the power efficiency typically falls within the range of 50 Watts per GHz to 80 Watts per GHz. In this analysis, we assume a representative value of 50 Watts per GHz to estimate the CPU energy consumption of the algorithm.

\end{document}